\documentclass[11pt,envcountsame,envcountsect,runningheads]{llncs}

\newcommand{\LongVersion}[1]{}

\usepackage[a4paper,margin=2.9cm,textheight=19cm]{geometry}

\usepackage{latexsym}
\usepackage{amssymb}
\usepackage{stmaryrd}
\usepackage{url}
\input{xy}
\xyoption{all}
\usepackage{xspace}
\usepackage[ruled,vlined,linesnumbered]{algorithm2e}

\renewcommand\marginpar[1]{}

\newcommand{\V}{\forall}

\newcommand{\ovl}[1]{\overline{#1}}
\newcommand{\Var}{\mathit{Vars}}

\def\trojkat{\mbox{{\scriptsize$\!\vartriangleleft$}}}
\newcommand{\myend}{\mbox{}\hfill\trojkat}

\newcommand{\edb}{extensional\xspace}
\newcommand{\idb}{intensional\xspace}

\newcommand{\comment}[1]{}
\def\eqref#1{(\ref{#1})}

\newcommand{\tpc}{\ovl{t}} % classical tuple
\newcommand{\tuple}[1]{(#1)}

\newcommand{\tuples}{\mathit{tuples}}
\newcommand{\tuplePairs}{\mathit{tuple\!\_pairs}}
\newcommand{\unprocessed}{\mathit{unprocessed}}
\newcommand{\unprocessedSubqueries}{\mathit{unprocessed\!\_subqueries}}
\newcommand{\unprocessedSQ}{\mathit{unprocessed\!\_subqueries_2}}
\newcommand{\unprocessedTuples}{\mathit{unprocessed\!\_tuples}}
\newcommand{\atom}{\mathit{atom}}
\newcommand{\subqueries}{\mathit{subqueries}}
\newcommand{\kind}{\mathit{kind}}
\newcommand{\preVars}{\mathit{pre\!\_vars}}
\newcommand{\postVars}{\mathit{post\!\_vars}}
\newcommand{\pred}{\mathit{pred}}

\newcommand{\inp}[1]{\mathit{input\!\_}#1}
\newcommand{\ans}[1]{\mathit{ans\!\_}#1}
\newcommand{\preFilter}{\mathit{pre\!\_filter}}
\newcommand{\postFilter}{\mathit{post\!\_filter}}
\newcommand{\filter}{\mathit{filter}}

\SetKwInput{GlobalData}{Global data}
\SetKwInput{Input}{Input}
\SetKwInput{Output}{Output}
\SetKwInput{Purpose}{Purpose}
\SetKw{Call}{call}
\SetKw{Set}{set}
\SetKw{Exit}{exit}
\SetKw{Break}{break}
\SetKwFunction{Transfer}{transfer}
\SetKwFunction{TransferS}{transfer2}
\SetKwFunction{Fire}{fire}
\SetKwFunction{FireS}{fire2}
\SetKwFunction{ActiveEdge}{active-edge}
\SetKwFunction{AddSubquery}{add-subquery}
\SetKwFunction{AddTuple}{add-tuple}
\SetKwFunction{AddTuplePair}{add-tuple-pair}
\SetKwFunction{ComputeGamma}{compute-gamma}
\SetKwFunction{termDepth}{term-depth}

\newcommand{\QSQTRE}{QSQTRE}

\title{Query-Subquery Nets}

\author{Linh Anh Nguyen\inst{1} \and Son Thanh Cao\inst{2}}
\institute{
Institute of Informatics, University of Warsaw\\
Banacha 2, 02-097 Warsaw, Poland\\
\email{nguyen@mimuw.edu.pl}
\and 
Faculty of Information Technology, Vinh University\\
182 Le Duan street, Vinh, Nghe An, Vietnam\\
\email{sonct@vinhuni.edu.vn}
}

\authorrunning{L.A. Nguyen and S.T. Cao}	

\begin{document}
\maketitle
\sloppy

\begin{abstract}
We formulate query-subquery nets and use them to create the first framework for developing algorithms for evaluating queries to Horn knowledge bases with the properties that: 
the approach is goal-directed; 
each subquery is processed only once and 
each supplement tuple, if desired, is transferred only once;
operations are done set-at-a-time;  
and any control strategy can be used. 
Our intention is to increase efficiency of query processing by eliminating redundant computation, increasing flexibility and reducing the number of accesses to the secondary storage. 
The framework forms a generic evaluation method called QSQN. 
To deal with function symbols, we use a term-depth bound for atoms and substitutions occurring in the computation and propose to use iterative deepening search which iteratively increases the term-depth bound. 
We prove soundness and completeness of our generic evaluation method and show that, when the term-depth bound is fixed, the method has PTIME data complexity.
We also present how tail recursion elimination can be incorporated into our framework and propose two exemplary control strategies, one is to reduce the number of accesses to the secondary storage, while the other is depth-first search. 

\medskip

\end{abstract}

%-------------------------------------------------------------------------------

\section{Introduction}
\label{section: intro}

Horn knowledge bases are definite logic programs, which are usually so big that either they cannot be totally loaded into the computer memory or evaluations for them cannot be done totally in the computer memory. Thus, in contrast to logic programming, for Horn knowledge bases efficient access to the secondary storage is an important aspect. Horn knowledge bases can be treated as extensions of Datalog deductive databases without the range-restrictedness and function-free conditions. 

This work studies query processing for Horn knowledge bases. It is a continuation of Madali{\'n}ska-Bugaj and Nguyen's work~\cite{ToCL455}. As argued in~\cite{ToCL455}, the Horn fragment of first-order logic plays an important role in knowledge representation and reasoning. The QSQN (query-subquery net) evaluation method provided in the current paper is essentially different from the QSQR (query-subquery recursive) method of~\cite{ToCL455}. However, some introductory and preliminary texts are borrowed from~\cite{ToCL455}.

An efficient method for evaluating queries to Horn knowledge bases should:
\begin{itemize}
\item be goal-directed, i.e.\ the computation should be closely related to the given goal
\item be set-oriented (instead of tuple-oriented) in order to reduce the number of accesses to the secondary storage 
\item do no redundant computation (or do it as less as possible).
\end{itemize} 

As discussed in~\cite{ToCL455}, to develop evaluation procedures for Horn knowledge bases one can either adapt tabled SLD-resolution systems of logic programming to reduce the number of accesses to the secondary storage or generalize evaluation methods of Datalog to deal with non-range-restricted definite logic programs and goals that may contain function symbols. 

Tabled SLD-resolution systems like OLDT \cite{OLDT}, SLD-AL~\cite{Vie87,Vie89}, linear tabulated resolution~\cite{LinTab1,LinTab2} are efficient computational procedures for logic programming without redundant recomputations, but they are not directly applicable to Horn knowledge bases to obtain efficient evaluation engines because they are not set-oriented (set-at-a-time). In particular, the suspension-resumption mechanism and the stack-wise representation as well as the ``global optimizations of SLD-AL'' are all tuple-oriented (tuple-at-a-time). Data structures for them are too complex so that they must be dropped if one wants to convert the methods to efficient set-oriented ones. The try of converting XSB~\cite{XSB,SagonasSW94} (a state-of-the-art implementation of OLDT) to Breadth-First XSB \cite{FreireSW97} as a set-oriented engine \cite{FreireSW97} for Horn knowledge bases removes essential features of XSB. Besides, as shown in Example~\ref{example1}, the breadth-first approach is not always efficient.

As well-known evaluation methods for Datalog deductive databases, there are the top-down methods QSQR~\cite{Vie86}, QoSaQ~\cite{Vie89}, QSQ~\cite{Vie86,Vie89,AHV95} and the bottom-up method based on magic-set transformation and seminaive evaluation~\cite{Magic1,RLK86,AHV95}. 
As the QSQ approach (including QSQR and QoSaQ) is based on SLD-resolution and the magic-set technique simulates QSQ, all of the mentioned methods are goal-directed. 

The first version of the QSQR (query-subquery recursive) evaluation method was formulated by Vieille in~\cite{Vie86} for Datalog deductive databases. It is set-oriented and uses a tabulation technique. That version is incomplete~\cite{Vie89,Nejdl87}. As pointed out by Mohamed Yahya~\cite{ToCL455}, the version given in the book~\cite{AHV95} by Abiteboul et al.\ is also incomplete. In~\cite{ToCL455}, Madali{\'n}ska-Bugaj and Nguyen corrected and generalized the method for Horn knowledge bases. The correction depends on clearing global $input\!\_$ relations for each iteration of the main loop. As observed by Vieille~\cite{Vie89}, the QSQR approach is like iterative deepening search. It allows redundant recomputations (see~\cite[Remark~3.2]{ToCL455}). 

The QoSaQ evaluation method~\cite{Vie89} is Vieille's adaptation of SLD-AL resolution for Datalog deductive databases. This evaluation method can be implemented as a set-oriented procedure, but Vieille stated that {\em ``We would like, however, to go even further and to claim that the practical interest of our approach lies in its one-inference-at-a-time basis, as opposed to having a set-theoretic basis. First, this tuple-based computational model permits a fine analysis of the duplicate elimination issue. \ldots''}~\cite[page 5]{Vie89}. Moreover, the specific techniques of QoSaQ like ``instantiation pattern'', ``rule compilation'', ``projection'' are heavily based on the range-restrictedness and function-free conditions. 

The magic-set technique~\cite{Magic1,RLK86} for Datalog deductive databases simulates the top-down QSQR evaluation by rewriting a given query to another equivalent one that when evaluated using a bottom-up technique (e.g.\ the seminaive evaluation) produces only facts produced by the QSQR evaluation. Some authors have extended the magic-set technique for Horn knowledge bases \cite{ramakrishnan92efficient,FreireSW97}. The bottom-up techniques usually use breadth-first search, and as shown in Example~\ref{example1}, are not always efficient. The magic-set transformation does not help for the case of that example. 

\begin{example} \label{example1}
The order of program clauses and the order of atoms in the bodies of program clauses may be essential, e.g., when the positive logic program defining intensional predicates is specified using the Prolog programming style. In such cases, the top-down depth-first approach may be much more efficient than the breadth-first approaches (including the one based on magic-set transformation and bottom-up seminaive evaluation). Here is such an example, in which $x$, $y$, $z$ denote variables and $a_i$, $b_{i,j}$ denote constant symbols: 
\begin{itemize}
\item the positive logic program (for defining intensional predicates $p$, $q_1$ and $q_2$):
  \[
  \begin{array}{l}
  p \gets q_1(a_0,a_{1000}) \\
  p \gets q_2(a_0,a_{1000}) \\[1ex]
  q_1(x,y) \gets r_1(x,y) \\
  q_1(x,y) \gets r_1(x,z), q_1(z,y) \\[1ex]
  q_2(x,y) \gets r_2(x,y) \\
  q_2(x,y) \gets r_2(x,z), q_2(z,y) 
  \end{array}
  \]

\item the extensional instance (for specifying extensional predicates $r_1$ and $r_2$): 
  \begin{eqnarray*}
  I(r_1) & = & \{(a_i, a_{i+1}) \mid 0 \leq i < 1000\} \\[1ex]
  I(r_2) & = & \{(a_0, b_{1,j}) \mid 1 \leq j \leq 1000\}\, \cup \\
         &   & \{(b_{i,j},b_{i+1,j}) \mid 1 \leq i < 999 \textrm{ and } 1 \leq j \leq 1000\}\, \cup \\
         &   & \{(b_{999,j},a_{1000}) \mid 1 \leq j \leq 1000\}
  \end{eqnarray*}
i.e.,
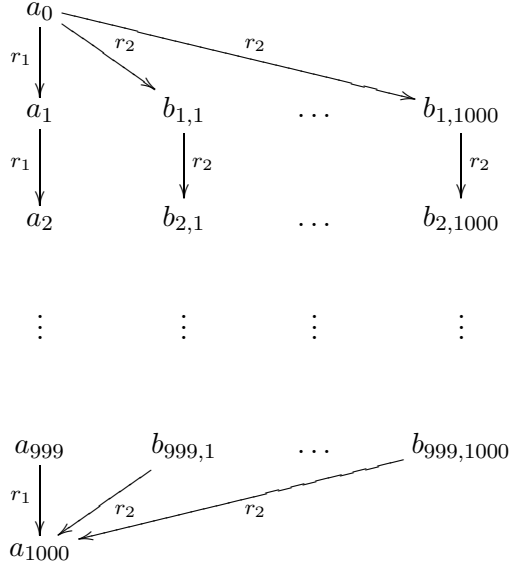
\begin{center}
\begin{tabular}{l}
\xymatrix{
a_0
\ar[d]_{r_1}
\ar[dr]^{r_2}
\ar[drrr]^{r_2}
\\
a_1
\ar[d]_{r_1}
&
b_{1,1}
\ar[d]^{r_2}
&
\ldots
& 
b_{1,1000}
\ar[d]^{r_2}
\\
a_2
&
b_{2,1}
&
\ldots
& 
b_{2,1000}
\\
\vdots
&
\vdots
&
\vdots
&
\vdots
\\
a_{999}
\ar[d]_{r_1}
&
b_{999,1}
\ar[dl]^{r_2}
&
\ldots
& 
b_{999,1000}
\ar[dlll]^{r_2}
\\
a_{1000}
}
\\
\ 
\end{tabular}
\end{center}
\item the goal: $\gets p$.
\end{itemize}

Our postulate is that the breadth-first approaches (including the evaluation method based on magic-set transformation and bottom-up seminaive evaluation) are too inflexible and not always efficient. Of course, depth-first search is not always good either. 
\myend
\end{example}

The QSQ (query-subquery) approach for Datalog queries, as presented in~\cite{AHV95}, originates from the QSQR method but allows a variety of control strategies. The QSQ framework~\cite{Vie86,AHV95} uses adornments to simulate SLD-resolution in pushing constant symbols from goals to subgoals. The annotated version of QSQ also uses annotations to simulate SLD-resolution in pushing repeats of variables from goals to subgoals (see~\cite{AHV95}).     

In this paper we generalize the QSQ approach for Horn knowledge bases. 
We formulate query-subquery nets and use them to create the first framework for developing algorithms for evaluating queries to Horn knowledge bases with the following properties: 
\begin{itemize}
\item the approach is goal-directed 
\item each subquery is processed only once
\item each supplement tuple, if desired, is transferred only once 
\item operations are done set-at-a-time 
\item any control strategy can be used. 
\end{itemize}
Our intention is to increase efficiency of query processing by eliminating redundant computation, increasing flexibility and reducing the number of accesses to the secondary storage. 
The framework forms a generic evaluation method called QSQN. 
Similarly to~\cite{ToCL455} but in contrast to the QSQ framework for Datalog queries~\cite{AHV95}, it does not use adornments and annotations (but has the effects of the annotated version). 
To deal with function symbols, we use a term-depth bound for atoms and substitutions occurring in the computation and propose to use iterative deepening search which iteratively increases the term-depth bound. 
We prove soundness and completeness of our generic evaluation method and show that, when the term-depth bound is fixed, the method has PTIME data complexity.
We also present how tail recursion elimination~\cite{Ross96} can be incorporated into our framework and propose two exemplary control strategies, one is to reduce the number of accesses to the secondary storage, while the other is depth-first search. 

The rest of this paper is structured as follows. 
In Section~\ref{section: prel} we recall some notions of first-order logic, logic programming, and Horn knowledge bases. In Section~\ref{section: QSQN} we present our QSQN evaluation method for Horn knowledge bases. We prove its soundness and completeness in Section~\ref{section: sound-comp} and estimate its data complexity in Section~\ref{section: complexity}. We consider tail recursion elimination in Section~\ref{section: TRE} and propose exemplary control strategies for our method in Section~\ref{section: strategies}. Concluding remarks are given in Section~\ref{section: conc}. 

%-------------------------------------------------------------------------------

\section{Preliminaries}
\label{section: prel}

First-order logic is considered in this work and we assume that the reader is familiar with it. We recall only the most important definitions for our work and refer the reader to \cite{Lloyd87,Apt97Book} for further reading. 

A signature for first-order logic consists of constant symbols, function symbols, and predicate symbols. Terms and formulas over a fixed signature are defined using the symbols of the signature and variables in the usual way. An {\em atom} is a formula of the form $p(t_1,\ldots,t_n)$, where $p$ is an $n$-ary predicate and $t_1,\ldots,t_n$ are terms. 
An {\em expression} is either a term, a tuple of terms, a formula without quantifiers or a list of formulas without quantifiers. A {\em simple expression} is either a term or an atom. 
The {\em term-depth of an expression} is the maximal nesting depth of function symbols occurring in that expression. 

%-------------------------------------------------------------------------------
\subsection{Substitution and Unification}

A {\em substitution} is a finite set $\theta = \{x_1/t_1, \ldots, x_k/t_k\}$, where $x_1, \ldots, x_k$ are pairwise distinct variables, $t_1, \ldots, t_k$ are terms, and $t_i \neq x_i$ for all $1 \leq i \leq k$. The set $dom(\theta) = \{x_1,\ldots,x_k\}$ is called the {\em domain} of $\theta$, while the set $range(\theta) = \{t_1,\ldots,t_k\}$ is called the {\em range} of $\theta$. 
By $\varepsilon$ we denote the {\em empty substitution}. 
The restriction of a substitution $\theta$ to a set $X$ of variables is the substitution $\theta_{|X} = \{(x/t) \in \theta \mid x \in X\}$. 
The {\em term-depth of a substitution} is the maximal nesting depth of function symbols occurring in that substitution. 

Let $\theta = \{x_1/t_1, \ldots, x_k/t_k\}$ be a substitution and $E$ be an expression. Then $E\theta$, the {\em instance} of $E$ by $\theta$, is the expression obtained from $E$ by simultaneously replacing all occurrences of the variable $x_i$ in $E$ by the term $t_i$, for $1 \leq i \leq k$.

Let $\theta = \{x_1/t_1, \ldots, x_k/t_k\}$  and $\delta = \{y_1/s_1, \ldots, y_h/s_h\}$ be substitutions (where $x_1,\ldots,x_k$ are pairwise distinct variables, and $y_1,\ldots,y_h$ are also pairwise distinct variables). Then the {\em composition} $\theta\delta$ of $\theta$ and $\delta$ is the substitution obtained from the sequence $\{x_1/(t_1\delta),\ldots,x_k/(t_k\delta), y_1/s_1,\ldots, y_h/s_h\}$ by deleting any binding $x_i/(t_i\delta)$ for which $x_i = (t_i\delta)$ and deleting any binding $y_j/s_j$ for which $y_j \in \{x_1,\ldots,x_k\}$.

A substitution $\theta$ is {\em idempotent} if $\theta\theta = \theta$. It is known that $\theta = \{x_1/t_1, \ldots, x_k/t_k\}$ is idempotent if none of $x_1,\ldots,x_k$ occurs in any $t_1,\ldots,t_k$. 

If $\theta$ and $\delta$ are substitutions such that $\theta\delta = \delta\theta = \varepsilon$, then we call them {\em renaming substitutions}. %and use $\theta^{-1}$ to denote $\delta$ (which is unique w.r.t.~$\theta$).
We say that an expression $E$ is a {\em variant} of an expression $E'$ if there exist substitutions $\theta$ and $\gamma$ such that $E = E'\theta$ and $E' = E\gamma$.

A substitution $\theta$ is {\em more general} than a substitution $\delta$ if there exists a substitution $\gamma$ such that $\delta = \theta\gamma$. Note that according to this definition, $\theta$ is more general than itself.

Let $\Gamma$ be a set of simple expressions. A substitution $\theta$ is called a {\em unifier} for $\Gamma$ if $\Gamma\theta$ is a singleton. If $\Gamma\theta = \{\varphi\}$ then we say that $\theta$ unifies $\Gamma$ (into $\varphi$). A unifier $\theta$ for $\Gamma$ is called a {\em most general unifier} (mgu) for $\Gamma$ if $\theta$ is more general than every unifier of $\Gamma$.

There is an effective algorithm, called the {\em unification algorithm}, for checking whether a set $\Gamma$ of simple expressions is unifiable (i.e.~has a unifier) and computing an idempotent mgu for $\Gamma$ if $\Gamma$ is unifiable (see, e.g., \cite{Lloyd87}).

%-------------------------------------------------------------------------------

If $E$ is an expression or a substitution then by $\Var(E)$ we denote the set of variables occurring in~$E$.
If $\varphi$ is a formula then by $\V (\varphi)$ we denote the {\em universal closure} of $\varphi$, which is the formula obtained by adding a universal quantifier for every variable having a free occurrence in $\varphi$.
%Similarly, $\E (\varphi)$ denotes the {\em existential closure} of $\varphi$, which is obtained by adding an existential quantifier for every variable having a free occurrence in~$\varphi$.

\subsection{Positive Logic Programs and SLD-Resolution}

A (positive or definite) {\em program clause} is a formula of the form $\V (A \lor \lnot B_1 \lor \ldots \lor \lnot B_k)$ with $k \geq 0$, written as \mbox{$A \gets B_1, \ldots, B_k$}, where $A$, $B_1$, \ldots, $B_k$ are atoms. $A$ is called the {\em head}, and $(B_1, \ldots, B_k)$ the {\em body} of the program clause. If $p$ is the predicate of $A$ then the program clause is called a program clause defining~$p$.

A {\em positive} (or {\em definite}) {\em logic program} is a finite set of program clauses.

A {\em goal} (also called a {\em negative clause}) is a formula of the form $\V (\lnot B_1 \lor \ldots \lor \lnot B_k)$, written as $\gets B_1, \ldots, B_k$, where $B_1, \ldots, B_k$ are atoms. If $k = 1$ then the goal is called a {\em unary goal}. If $k = 0$ then the goal stands for falsity and is called the {\em empty goal} (or the {\em empty clause}) and denoted by~$\Box$.

If $P$ is a positive logic program and $G =\ \gets B_1,\ldots,B_k$ is a goal, then $\theta$ is called a {\em correct answer} for $P \cup \{G\}$ if $P \models \V((B_1 \land\ldots\land B_k)\theta)$.

We now give definitions for SLD-resolution.

A goal $G'$ is {\em derived} from a goal $G =\ \gets A_1,\ldots,A_i,
 \ldots, A_k$ and a program clause $\varphi = (A \gets B_1,\ldots, B_h)$ using $A_i$ as the {\em selected atom} and $\theta$ as the most general unifier (mgu) if $\theta$ is an mgu for $A_i$ and $A$, and $G' =\ \gets (A_1,\ldots,A_{i-1},B_1,\ldots,B_h,A_{i+1},\ldots,A_k)\theta$. We call $G'$ a {\em resolvent} of $G$ and~$\varphi$. If $i = 1$ then we say that $G'$ is derived from $G$ and $\varphi$ using {\em the leftmost selection function}.

Let $P$ be a positive logic program and $G$ be a goal.

An {\em SLD-derivation} from $P \cup \{G\}$ consists of a (finite or infinite) sequence $G_0 = G$, $G_1$, $G_2$, \ldots of goals, a sequence $\varphi_1, \varphi_2, \ldots$ of variants of program clauses of $P$ and a sequence $\theta_1, \theta_2, \ldots$ of mgu's such that each $G_{i+1}$ is derived from $G_i$ and $\varphi_{i+1}$ using~$\theta_{i+1}$.
Each $\varphi_i$ is a suitable variant of the corresponding
program clause. That is, $\varphi_i$ does not have any variables
which already appear in the derivation up to $G_{i-1}$. Each
program clause variant $\varphi_i$ is called an {\em input program
clause}.

An {\em SLD-refutation} of $P \cup \{G\}$ is a finite SLD-derivation from $P \cup \{G\}$ which has the empty clause as the last goal in the derivation.

A {\em computed answer} $\theta$ for $P \cup \{G\}$ is the substitution obtained by restricting the composition $\theta_1\ldots\theta_n$ to the variables of $G$, where $\theta_1,\ldots,\theta_n$ is the sequence of mgu's occurring in an SLD-refutation of $P \cup \{G\}$.

\begin{theorem}[Soundness and Completeness of SLD-Resolution \cite{Clark,Stark89}]
\label{theorem: SLD soundness and completeness}
Let $P$ be a positive logic program and $G$ be a goal.
Then every computed answer for $P \cup \{G\}$ is a correct answer for $P \cup \{G\}$.
Conversely, for every correct answer $\theta$ for $P \cup \{G\}$, using any selection function there exists a computed answer $\delta$ for $P \cup \{G\}$ such that $G\theta = G\delta\gamma$ for some substitution~$\gamma$.
\myend
\end{theorem}

We will use also the following well-known lemmas:

\begin{lemma}[Lifting Lemma]
\label{lifting lemma}
%Let $P$ be a positive logic program, $G$ be a goal, $\theta$ be a substitution, and $l$ be a natural number. Suppose there exists an SLD-refutation of $P \cup \{G\theta\}$ using mgu's $\theta_1,\ldots,\theta_n$ such that the variables of the input program clauses are distinct from the variables in $G$ and $\theta$ and the term-depths of the goals and the composition $\theta_1\ldots\theta_n$ are bounded by~$l$. Then there exist a substitution $\gamma$ and an SLD-refutation of $P \cup \{G\}$ using the same sequence of input program clauses, the same selected atoms and mgu's $\theta'_1,\ldots,\theta'_n$ such that the term-depths of the goals and the composition $\theta'_1\ldots\theta'_n$ are bounded by $l$ and $\theta\theta_1\ldots\theta_n = \theta'_1\ldots\theta'_n\gamma$.
Let $P$ be a positive logic program, $G$ be a goal, $\theta$ be a substitution, and $l$ be a natural number. Suppose there exists an SLD-refutation of $P \cup \{G\theta\}$ using mgu's $\theta_1,\ldots,\theta_n$ such that the variables of the input program clauses are distinct from the variables in $G$ and $\theta$ and the term-depths of the goals are bounded by~$l$. Then there exist a substitution $\gamma$ and an SLD-refutation of $P \cup \{G\}$ using the same sequence of input program clauses, the same selected atoms and mgu's $\theta'_1,\ldots,\theta'_n$ such that the term-depths of the goals are bounded by $l$ and $\theta\theta_1\ldots\theta_n = \theta'_1\ldots\theta'_n\gamma$.
\myend
\end{lemma}

The Lifting Lemma given in \cite{Lloyd87} does not contain the condition ``the variables of the input program clauses are distinct from the variables in $G$ and $\theta$'' and is therefore inaccurate (see, e.g., \cite{Apt97Book}). The correct version given above follows from the one presented, amongst others, in~\cite{StaabCourse}. For applications of this lemma in this paper, we assume that {\em fresh variables} from a special infinite list of variables are used for renaming variables of input program clauses in SLD-derivations, and that mgu's are computed using a standard method. The mentioned condition will thus be satisfied.  

In a computational process, a {\em fresh variant} of a formula $\varphi$, where $\varphi$ can be an atom, a goal $\gets A$ or a program clause $A \gets B_1,\ldots,B_k$ (written without quantifiers), is a formula $\varphi\theta$, where $\theta$ is a renaming substitution such that $dom(\theta) = \Var(\varphi)$ and $range(\theta)$ consists of fresh variables that were not used in the computation (and the input). 

%-------------------------------------------------------------------------------

\subsection{Definitions for Horn Knowledge Bases}

Similarly as for deductive databases, we classify each predicate either as {\em intensional} or as {\em extensional}.  
A {\em generalized tuple} is a tuple of terms, which may contain function symbols and variables.
A {\em generalized relation} is a set of generalized tuples of the same arity.
A {\em Horn knowledge base} is defined to be a pair consisting of a positive logic program for defining intensional predicates and a {\em generalized \edb\ instance}, which is a function mapping each \edb\ $n$-ary predicate to an $n$-ary generalized relation. Note that intensional predicates are defined by a positive logic program which may contain function symbols and not be range-restricted.
From now on, we use the term ``relation'' to mean a generalized relation, and the term ``\edb\ instance'' to mean a generalized \edb\ instance.

%\textbf{Note:} We will treat a tuple $\tpc$ from a relation associated with a predicate $p$ as the atom $p(\tpc)$. Thus, a relation (of tuples) of a predicate $p$ is a set of atoms of $p$, and an \edb\ instance is a set of atoms of \edb\ predicates. Conversely, a set of atoms of $p$ can be treated as a relation (of tuples) of the predicate~$p$.

Given a Horn knowledge base specified by a positive logic program $P$ and an \edb\ instance $I$, a {\em query} to the knowledge base is a positive formula $\varphi(\ovl{x})$ without quantifiers, where $\ovl{x}$ is a tuple of all the variables of $\varphi$.\footnote{A {\em positive formula without quantifiers} is a formula built up from atoms using only connectives $\land$ and~$\lor$.} A ({\em correct}) {\em answer} for the query is a tuple $\tpc$ of terms of the same length as $\ovl{x}$ such that $P \cup I \models \V(\varphi(\tpc))$. When measuring {\em data complexity}, we assume that $P$ and $\varphi$ are fixed, while $I$ varies. Thus, the pair $(P,\varphi(\ovl{x}))$ is treated as a {\em query} to the \edb\ instance $I$. We will use the term ``query'' in this meaning.

It can easily be shown that, every query $(P,\varphi(\ovl{x}))$ can be
transformed in polynomial time to an equivalent query of the
form $(P',q(\ovl{x}))$ over a signature extended with new \idb\
predicates, including $q$. The equivalence means that, for every
\edb\ instance $I$ and every tuple $\tpc$ of terms of the same
length as $\ovl{x}$, $P \cup I \models \V(\varphi(\tpc))$ iff $P'
\cup I \models \V(q(\tpc))$. The transformation is based on
introducing new predicates for defining complex subformulas
occurring in the query. For example, if $\varphi = p(x) \land
r(x,y)$, then \mbox{$P' = P \cup \{q(x,y) \gets p(x), r(x,y)\}$}, where
$q$ is a new \idb\ predicate.

Without loss of generality, we will consider only queries of the form $(P,q(\ovl{x}))$, where $q$ is an \idb\ predicate. Answering such a query on an \edb\ instance $I$ is to find (correct) answers for $P \cup I \cup \{\gets q(\ovl{x})\}$.

%-------------------------------------------------------------------------------

\section{Query-Subquery Nets}
\label{section: QSQN}

Let $P$ be a positive logic program and $\varphi_1$, \ldots, $\varphi_m$ be all the program clauses of $P$, with
\[ \varphi_i = (A_i \gets B_{i,1}, \ldots, B_{i,n_i})\] 
where $n_i \geq 0$. A {\em query-subquery net structure} (in short, {\em QSQ-net structure}) of $P$ is a tuple $\tuple{V,E,T}$ such that:
\begin{itemize}
\item $V$ consists of nodes
   \begin{itemize}
   \item $input\!\_p$ and $ans\!\_p$ for each intensional predicate $p$ of $P$
   \item $\preFilter_i$, $\filter_{i,1}$, \ldots, $\filter_{i,n_i}$, $\postFilter_i$ for each $1 \leq i \leq m$
   \end{itemize}
\item $E$ consists of edges 
   \begin{itemize}
   \item $(\filter_{i,1}, \filter_{i,2})$, \ldots, $(\filter_{i,n_i-1}$, $\filter_{i,n_i})$ for each $1 \leq i \leq m$
   \item $(\preFilter_i, \filter_{i,1})$ and $(\filter_{i,n_i}, \postFilter_i)$ for each $1 \leq i \leq m$ with $n_i \geq 1$
   \item $(\preFilter_i, \postFilter_i)$ for each $1 \leq i \leq m$ with $n_i = 0$
   \item $(\inp{p},\preFilter_i)$ and $(\postFilter_i,\ans{p})$ for each $1 \leq i \leq m$, where $p$ is the predicate of~$A_i$
   \item $(\filter_{i,j},\inp{p})$ and $(\ans{p},\filter_{i,j})$ for each intensional predicate $p$ and each \mbox{$1 \leq i \leq m$} and $1 \leq j \leq n_i$ such that $B_{i,j}$ is an atom of~$p$
   \end{itemize}

\item $T$ is a function, called the {\em memorizing type} of the net structure, mapping each node $\filter_{i,j} \in V$ such that the predicate of $B_{i,j}$ is extensional to {\em true} or {\em false}.
\end{itemize}
If $(v,w) \in E$ then we call $w$ a {\em successor} of $v$, and $v$ a {\em predecessor} of $w$. Note that $V$ and $E$ are uniquely specified by $P$. We call the pair $\tuple{V,E}$ the {\em QSQ topological structure} of~$P$. 

\begin{example}\label{example: HGDSA} 
Figure~\ref{fig: QSQ-net example} illustrates the QSQ topological structure of the following positive logic program:
\[
\begin{array}{l}
p(x,y) \gets q(x,y) \\
p(x,y) \gets q(x,z), p(z,y).
\end{array}
\]
\myend
\end{example}

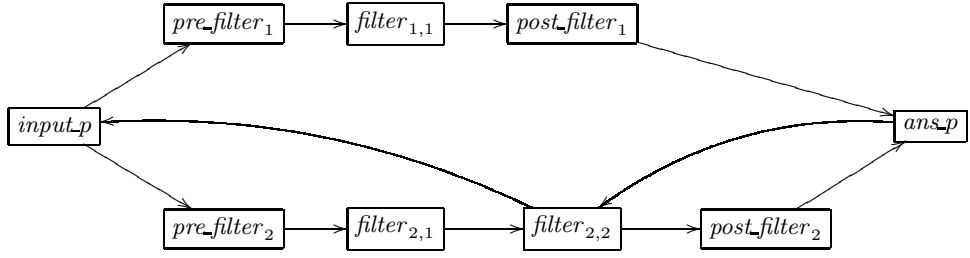
\begin{figure}
\begin{footnotesize}
\begin{center}
\begin{tabular}{c}
\xymatrix{
&
*+[F]{\preFilter_1}
\ar@{->}[r]
&
*+[F]{\filter_{1,1}}
\ar@{->}[r]
&
*+[F]{\postFilter_1}
\ar@{->}[rrd]
\\
*+[F]{\inp{p}}
\ar@{->}[ru]
\ar@{->}[rd]
& & & & & 
*+[F]{\ans{p}}
\ar@/_{1.3pc}/@{->}[lld]
\\
&
*+[F]{\preFilter_2}
\ar@{->}[r]
&
*+[F]{\filter_{2,1}}
\ar@{->}[r]
&
*+[F]{\filter_{2,2}}
\ar@{->}[r]
\ar@/_{1.3pc}/@{->}[lllu]
&
*+[F]{\postFilter_2}
\ar@{->}[ru]
} % \xymatrix
\end{tabular}
\end{center}
\end{footnotesize}
\caption{The QSQ topological structure of the program given in Example~\ref{example: HGDSA}.
\label{fig: QSQ-net example}
}
\end{figure}

A {\em query-subquery net} (in short, {\em QSQ-net}) of $P$ is a tuple $N = \tuple{V,E,T,C}$ such that $\tuple{V,E,T}$ is a QSQ-net structure of $P$ and $C$ is a mapping that associates each node $v \in V$ with a structure called the {\em contents} of $v$, satisfying the following conditions:
\begin{itemize}
\item $C(v)$, where $v = \inp{p}$ or $v = \ans{p}$ for an intensional predicate $p$ of $P$, consists of:
   \begin{itemize}
   %\item $\tuples(v)\,$: a set of generalized tuples of the same arity as $p$ such that, if different tuples $\tpc_1$ and $\tpc_2$ belong to $\tuples(v)$, then none of them is an instance of the other
   \item $\tuples(v)\,$: a set of generalized tuples of the same arity as $p$
   \item $\unprocessed(v,w)$ for $(v,w) \in E$: a subset of $\tuples(v)$
   \end{itemize}

\item $C(v)$, where $v = \preFilter_i$, consists of:
   \begin{itemize}
   \item $\atom(v) = A_i$ and $\postVars(v) = \Var((B_{i,1},\ldots,B_{i,n_i}))$
   \end{itemize}

\item $C(v)$, where $v = \postFilter_i$, is empty, but we assume $\preVars(v) = \emptyset$

\item $C(v)$, where $v = \filter_{i,j}$ and $p$ is the predicate of $B_{i,j}$, consists~of:
   \begin{itemize}
   \item $\kind(v) = \edb$ if $p$ is extensional, and $\kind(v) = \idb$ otherwise
   \item $\pred(v) = p$ and $\atom(v) = B_{i,j}$
   \item $\preVars(v)= \Var((B_{i,j},\ldots,B_{i,n_i}))$ and $\postVars(v)= \Var((B_{i,j+1},\ldots,B_{i,n_i}))$
   \item $\subqueries(v)$: a set of pairs of the form $(\tpc,\delta)$, where $\tpc$ is a generalized tuple of the same arity as the predicate of $A_i$ and $\delta$ is an idempotent substitution such that $dom(\delta) \subseteq \preVars(v)$ and $dom(\delta) \cap \Var(\tpc) = \emptyset$
   \item $\unprocessedSubqueries(v) \subseteq \subqueries(v)$
   \item in the case $p$ is intensional:
	\begin{itemize}
	   \item $\unprocessedSQ(v) \subseteq \subqueries(v)$
	   \item $\unprocessedTuples(v)\,$: a set of generalized tuples of the same arity as~$p$
	\end{itemize}
   \end{itemize} 

\item if $v = \filter_{i,j}$, $\kind(v) = \edb$ and $T(v) = false$ then $\subqueries(v) = \emptyset$.
\end{itemize}

By a {\em subquery} we mean a pair of the form $(\tpc,\delta)$, where $\tpc$ is a generalized tuple and $\delta$ is an idempotent substitution such that $dom(\delta) \cap \Var(\tpc) = \emptyset$.

For $v = \filter_{i,j}$ and $p$ being the predicate of $A_i$, the meaning of a subquery $(\tpc,\delta) \in subqueries(v)$ is that: for processing the goal $\gets p(\ovl{s})$ with $\ovl{s} \in \tuples(\inp{p})$ using the program clause $\varphi_i = (A_i \gets B_{i,1},\ldots,B_{i,n_i})$, unification of $p(\ovl{s})$ and $A_i$ as well as processing of the subgoals $B_{i,1},\ldots,B_{i,j-1}$ were done, amongst others, by using a sequence of mgu's $\gamma_0,\ldots,\gamma_{j-1}$ with the property that $\tpc = \ovl{s}\gamma_0\ldots\gamma_{j-1}$ and $\delta = (\gamma_0\ldots\gamma_{j-1})_{|\Var((B_{i,j},\ldots,B_{i,n_i}))}$. 

An {\em empty QSQ-net of $P$} is a QSQ-net of $P$ with all the sets of the
forms $\tuples(v)$, $\unprocessed(v,w)$, $\subqueries(v)$, $\unprocessedSubqueries(v)$, $\unprocessedSQ(v)$, $\unprocessedTuples(v)$ being empty. 

In a QSQ-net, if $v = \preFilter_i$ or $v = \postFilter_i$ or $v = \filter_{i,j}$ and $\kind(v) = \edb$ then $v$ has exactly one successor, which we denote by $succ(v)$. 

If $v$ is $\filter_{i,j}$ with $\kind(v) = \idb$ and $\pred(v) = p$ then $v$ has exactly two successors. In that case, let 
\[
succ(v) = \left\{
	\begin{array}{ll}
	\filter_{i,j+1} & \textrm{if } n_i > j\\
	\postFilter_i\quad & \textrm{otherwise}
	\end{array}
	\right.
\] 
and $succ_2(v) = \inp{p}$. The set $\unprocessedSubqueries(v)$ is used for (i.e.\ corresponds to) the edge $(v,succ(v))$, while $\unprocessedSQ(v)$ is used for the edge $(v,succ_2(v))$.  

Note that if $succ(v) = w$ then $\postVars(v) = \preVars(w)$. 
In particular, $\postVars(\filter_{i,n_i}) = \preVars(\postFilter_i) = \emptyset$. 

The formats of data transferred through edges of a QSQ-net are specified as follows:
\begin{itemize}
\item data transferred through an edge of the form $(\inp{p},v)$, $(v,\inp{p})$, $(v,\ans{p})$ or $(\ans{p},v)$ is a finite set of generalized tuples of the same arity as $p$
\item data transferred through an edge $(u,v)$ with $v = \filter_{i,j}$ and $u$ not being of the form $\ans{p}$ is a finite set of subqueries that can be added to $\subqueries(v)$
\item data transferred through an edge $(v, \postFilter_i)$ is a set of subqueries $(\tpc,\varepsilon)$ such that $\tpc$ is a generalized tuple of the same arity as the predicate of $A_i$.
\end{itemize}

If $(\tpc,\delta)$ and $(\tpc',\delta')$ are subqueries that can be transferred through an edge to $v$ then we say that $(\tpc,\delta)$ is {\em more general} than $(\tpc',\delta')$ w.r.t.~$v$, and that $(\tpc',\delta')$ is {\em less general} than $(\tpc,\delta)$ w.r.t.~$v$, if there exists a substitution $\gamma$ such that $\tpc\gamma = \tpc'$ and $(\delta\gamma)_{|\preVars(v)} = \delta'$.  

%-------------------------------------------------------------------------------

Informally, a subquery $(\tpc,\delta)$ transferred through an edge to $v$ is processed as follows:
\begin{itemize}
\item if $v = \filter_{i,j}$, $\kind(v) = \edb$ and $\pred(v) = p$ then, for each $\tpc' \in I(p)$, if $\atom(v)\delta = B_{i,j}\delta$ is unifiable with a fresh variant of $p(\tpc')$ by an mgu $\gamma$ then transfer the subquery $(\tpc\gamma,(\delta\gamma)_{|\postVars(v)})$ through $(v,succ(v))$

\item if $v = \filter_{i,j}$, $\kind(v) = \idb$ and $\pred(v) = p$ then 
  \begin{itemize}
  \item transfer the input tuple $\tpc'$ such that $p(\tpc') = \atom(v)\delta = B_{i,j}\delta$ through $(v,\inp{p})$ to add a fresh variant of it to $\tuples(\inp{p})$
  \item for each currently existing $\tpc' \in \tuples(\ans{p})$, if $\atom(v)\delta = B_{i,j}\delta$ is unifiable with a fresh variant of $p(\tpc')$ by an mgu $\gamma$ then transfer the subquery $(\tpc\gamma,(\delta\gamma)_{|\postVars(v)})$ through $(v,succ(v))$
  \item store the subquery $(\tpc,\delta)$ in $\subqueries(v)$, and later, for each new $\tpc'$ added to $\tuples(\ans{p})$, if $\atom(v)\delta = B_{i,j}\delta$ is unifiable with a fresh variant of $p(\tpc')$ by an mgu $\gamma$ then transfer the subquery $(\tpc\gamma,(\delta\gamma)_{|\postVars(v)})$ through $(v,succ(v))$
  \end{itemize}

\item if $v = \postFilter_i$ and $p$ is the predicate of $A_i$ then transfer the answer tuple $\tpc$ through $(postFilter_i,\ans{p})$ to add it to $\tuples(\ans{p})$. 
\end{itemize}

Formally, the processing of a subquery is designed more sophisticatedly so that:
\begin{itemize}
\item every subquery / input tuple / answer tuple subsumed by another one is ignored
\item every subquery / input tuple / answer tuple with term-depth greater than $l$ is ignored 
\item the processing is divided into smaller steps which can be delayed to maximize flexibility and allow various control strategies
\item the processing is done set-at-a-time (e.g., for all the unprocessed subqueries accumulated in a given node).
\end{itemize}

Procedure $\Transfer(D,u,v)$ (given on page~\pageref{proc: Transfer}) specifies the effects of transferring data~$D$ through an edge $(u,v)$ of a QSQ-net. 
If $v$ is of the form $\preFilter_i$ or $\postFilter_i$ or ($v = \filter_{i,j}$ and $\kind(v) = \edb$ and $T(v) = false$) then the input $D$ for $v$ is processed immediately and appropriate data $\Gamma$ is produced and transferred through $(v,succ(v))$. Otherwise, the input $D$ for $v$ is not processed immediately, but accumulated into the structure of $v$ in an appropriate way. 

Function $\ActiveEdge(u,v)$ (given on page~\pageref{func: ActiveEdge}) returns $true$ for an edge $(u,v)$ if data accumulated in $u$ can be processed to produce some data to transfer through $(u,v)$, and returns $false$ otherwise. 

In the case $\ActiveEdge(u,v)$ is true, procedure $\Fire(u,v)$ (given on page~\pageref{proc: Fire}) processes data accumulated in $u$ that has not been processed before to transfer appropriate data through the edge $(u,v)$. 

Algorithm~\ref{alg: QSQN} (given on page~\pageref{alg: QSQN}) presents our QSQN evaluation method for Horn knowledge bases. 

%-------------------------------------------------------------------------------

\begin{figure*}
\begin{procedure}[H]
\caption{add-subquery($\tpc,\delta,\Gamma,v$)\label{proc: AddSubquery}} 
\Purpose{add the subquery $(\tpc,\delta)$ to $\Gamma$, but keep in $\Gamma$ only the most general subqueries w.r.t.~$v$.}
\BlankLine

\If{$\termDepth(\tpc) \leq l$ and $\termDepth(\delta) \leq l$ and no subquery in $\Gamma$ is more general than $(\tpc,\delta)$ w.r.t.~$v$}{
   delete from $\Gamma$ all subqueries less general than $(\tpc,\delta)$ w.r.t.~$v$\;
   add $(\tpc,\delta)$ to $\Gamma$
}
\end{procedure}

\medskip

\begin{procedure}[H]
\caption{add-tuple($\tpc,\Gamma$)\label{proc: AddTuple}} 
\Purpose{add the tuple $\tpc$ to $\Gamma$, but keep in $\Gamma$ only the most general tuples.}
\BlankLine

let $\tpc'$ be a fresh variant of $\tpc$\;
\If{$\tpc'$ is not an instance of any tuple from $\Gamma$}{
   delete from $\Gamma$ all tuples that are instances of $\tpc'$\;
   add $\tpc'$ to $\Gamma$
}
\end{procedure}
\end{figure*}

\begin{procedure}
\caption{transfer($D,u,v$)\label{proc: Transfer}} 
\GlobalData{a Horn knowledge base $\tuple{P,I}$, a QSQ-net $N = \tuple{V,E,T,C}$ of $P$, and a~term-depth bound $l$.}
\Input{data $D$ to transfer through the edge $(u,v) \in E$.}
%\DontPrintSemicolon

\BlankLine

\lIf{$D = \emptyset$}{\Return}\;

%\BlankLine

\uIf{$u$ is $\inp{p}$}{
\label{transfer: step FDWES}
$\Gamma := \emptyset$\;
\ForEach{$\tpc \in D$}{
  \If{$p(\tpc)$ and $\atom(v)$ are unifiable by an mgu $\gamma$}{
	$\AddSubquery(\tpc\gamma, \gamma_{|\postVars(v)}, \Gamma, succ(v))$
  }
}
$\Transfer(\Gamma,v,succ(v))$
}
\lElseIf{$u$ is $\ans{p}$}{$\unprocessedTuples(v) := \unprocessedTuples(v) \cup D$}\\
\uElseIf{$v$ is $\inp{p}$ or $\ans{p}$}{
  \ForEach{$\tpc \in D$}{
     let $\tpc'$ be a fresh variant of $\tpc$\;
     \If{$\tpc'$ is not an instance of any tuple from $\tuples(v)$}{
	\ForEach{$\tpc'' \in \tuples(v)$}{
	   \If{$\tpc''$ is an instance of $\tpc'$}{
	      delete $\tpc''$ from $\tuples(v)$\;
	      \lForEach{$\tuple{v,w} \in E$}{delete $\tpc''$ from $\unprocessed(v,w)$}
	   }
	}
	\uIf{$v$ is $\inp{p}$}{
	   add $\tpc'$ to $\tuples(v)$\;
	   \lForEach{$\tuple{v,w} \in E$}{add $\tpc'$ to $\unprocessed(v,w)$}
	}
	\Else{
	   add $\tpc$ to $\tuples(v)$\;
	   \lForEach{$\tuple{v,w} \in E$}{add $\tpc$ to $\unprocessed(v,w)$}
	}
     }
  }
}
\uElseIf{$v$ is $\filter_{i,j}$ and $\kind(v) = \edb$ and $T(v) = false$}{
  let $p = \pred(v)$ and set $\Gamma := \emptyset$\;
  \ForEach{$(\tpc,\delta) \in D$}{
     \If{$\termDepth(\atom(v)\delta) \leq l$}{
        \ForEach{$\tpc' \in I(p)$}{
	 \If{$\atom(v)\delta$ is unifiable with a fresh variant of $p(\tpc')$ by an mgu $\gamma$}{
	 	$\AddSubquery(\tpc\gamma, (\delta\gamma)_{|\postVars(v)}, \Gamma, succ(v))$
	 }
	}
     }
  }
  $\Transfer(\Gamma,v,succ(v))$
}
\uElseIf{$v$ is $\filter_{i,j}$ and ($\kind(v) = \edb$ and $T(v) = true$ or $\kind(v) = \idb$)}{
  \ForEach{$(\tpc,\delta) \in D$}{
     \If{$\termDepth(\atom(v)\delta) \leq l$}{
        \If{no subquery in $\subqueries(v)$ is more general than $(\tpc,\delta)$}{
           delete from $\subqueries(v)$ all subqueries less general than $(\tpc,\delta)$\;
           delete from $\unprocessedSubqueries(v)$ all subqueries less general than $(\tpc,\delta)$\;
	   add $(\tpc,\delta)$ to both $\subqueries(v)$ and $\unprocessedSubqueries(v)$\;
	   \If{$\kind(v) = \idb$}{
              delete from $\unprocessedSQ(v)$ all subqueries less general than $(\tpc,\delta)$\;
	      add $(\tpc,\delta)$ to $\unprocessedSQ(v)$
	   }
        }
     }
  }
}
\Else(\tcp*[h]{$v$ is of the form $\postFilter_i$}){
  $\Gamma := \{\tpc \mid (\tpc,\varepsilon) \in D\}$\;
  $\Transfer(\Gamma,v,succ(v))$
  \label{transfer: step FDWET}
}
\end{procedure}

\begin{figure*}
\begin{function}[H]
\caption{active-edge($u,v$)\label{func: ActiveEdge}} 
\GlobalData{a QSQ-net $N = \tuple{V,E,T,C}$.}
\Input{an edge $(u,v) \in E$.}
\Output{$true$ if there are data to transfer through the edge $(u,v)$, and $false$ otherwise.}
%\DontPrintSemicolon

\BlankLine
\lIf{$u$ is $\preFilter_i$ or $\postFilter_i$}{\Return false}\\
\lElseIf{$u$ is $\inp{p}$ or $\ans{p}$}{\Return $\unprocessed(u,v) \neq \emptyset$}\\
\uElseIf{$u$ is $\filter_{i,j}$ and $\kind(u) = \edb$}{\Return $T(u) = true \land \unprocessedSubqueries(u) \neq \emptyset$}
\Else(\tcp*[h]{$u$ is of the form $\filter_{i,j}$ and $\kind(u) = \idb$}){
  let $p = \pred(u)$\;
  \lIf{$v = \inp{p}$}{\Return $\unprocessedSQ(u) \neq \emptyset$}\\
  \lElse{\Return $\unprocessedSubqueries(u) \neq \emptyset \lor \unprocessedTuples(u) \neq \emptyset$}
}
\end{function}

\medskip

\begin{procedure}[H]
\caption{fire($u,v$)\label{proc: Fire}} 
\GlobalData{a Horn knowledge base $\tuple{P,I}$, a QSQ-net $N = \tuple{V,E,T,C}$ of $P$, and a~term-depth bound $l$.}
\Input{an edge $(u,v) \in E$ such that $\ActiveEdge(u,v)$ holds.}
%\DontPrintSemicolon

\BlankLine

\uIf{$u$ is $\inp{p}$ or $\ans{p}$}{
   $\Transfer(\unprocessed(u,v),u,v)$\;
   $\unprocessed(u,v) := \emptyset$
}
\uElseIf{$u$ is $\filter_{i,j}$ and $\kind(u) = \edb$ and $T(u) = true$}{
  let $p = \pred(u)$ and set $\Gamma := \emptyset$\; 
  \ForEach{$(\tpc,\delta) \in \unprocessedSubqueries(u)$}{
     \ForEach{$\tpc' \in I(p)$}{
	\If{$\atom(u)\delta$ is unifiable with a fresh variant of $p(\tpc')$ by an mgu $\gamma$}{
	   $\AddSubquery(\tpc\gamma,(\delta\gamma)_{|\postVars(u)}, \Gamma,v)$
	}
     }
  }
  $\unprocessedSubqueries(u) := \emptyset$\;
  $\Transfer(\Gamma,u,v)$
}
\ElseIf{$u$ is $\filter_{i,j}$ and $\kind(u) = \idb$}{
  let $p = \pred(u)$ and set $\Gamma := \emptyset$\;
  \uIf{$v = \inp{p}$}{
     \lForEach{$(\tpc,\delta) \in \unprocessedSQ(u)$}{let $p(\tpc') = \atom(u)\delta$, $\AddTuple(\tpc',\Gamma)$}\label{fire: Step HGSAA}\;
     $\unprocessedSQ(u) := \emptyset$\;
  }
  \Else{
     \ForEach{$(\tpc,\delta) \in \unprocessedSubqueries(u)$}{
        \ForEach{$\tpc' \in \tuples(\ans{p})$}{
          \If{$\atom(u)\delta$ is unifiable with a fresh variant of $p(\tpc')$ by an mgu $\gamma$}{
	     $\AddSubquery(\tpc\gamma,(\delta\gamma)_{|\postVars(u)}, \Gamma, v)$
	  }
        }
     }
     $\unprocessedSubqueries(u) := \emptyset$\;
     \BlankLine
     \If{$\unprocessedTuples(u) \neq \emptyset$}{
       \ForEach{$\tpc \in \unprocessedTuples(u)$}{
          \ForEach{$(\tpc',\delta) \in \subqueries(u)$}{
            \If{$\atom(u)\delta$ is unifiable with a fresh variant of $p(\tpc)$ by an mgu $\gamma$}{
	      $\AddSubquery(\tpc'\gamma,(\delta\gamma)_{|\postVars(u)}, \Gamma, v)$
	    }
          }
       }
       $\unprocessedTuples(u) := \emptyset$
     }
  }
  $\Transfer(\Gamma,u,v)$
}
\end{procedure}
\end{figure*}

\begin{figure*}[t]
%\LinesNumberedHidden
\begin{algorithm}[H]
\caption{for evaluating a query $(P,q(\ovl{x}))$ on an \edb\ instance $I$.\label{alg: QSQN}}
%\DontPrintSemicolon

let $\tuple{V,E,T}$ be a QSQ-net structure of $P$\tcp*{$T$ can be chosen arbitrarily}

set $C$ so that $N = \tuple{V,E,T,C}$ is an empty QSQ-net of $P$\;

\BlankLine

let $\ovl{x}'$ be a fresh variant of $\ovl{x}$\;
$\tuples(\inp{q}) := \{\ovl{x}'\}$\;
\lForEach{$(\inp{q},v) \in E$}{$\unprocessed(\inp{q},v) := \{\ovl{x}'\}$}\;

\BlankLine

\While{there exists $(u,v) \in E$ such that $\ActiveEdge(u,v)$ holds}{
  select $(u,v) \in E$ such that $\ActiveEdge(u,v)$ holds\;
  \tcp{any strategy is acceptable for the above selection}
  $\Fire(u,v)$
}

\BlankLine
\Return $\tuples(\ans{q})$
\end{algorithm}
\end{figure*}

\subsection{Relaxing Term-Depth Bound}

Suppose that we want to compute as many as possible but no more than $k$ correct answers for a query $(P,q(\ovl{x}))$ on an \edb\ instance $I$ within time limit $L$. Then we can use iterative deepening search which iteratively increases term-depth bound for atoms and substitutions occurring in the computation as follows:
\begin{enumerate}
\item Initialize term-depth bound $l$ to 0 (or another small natural number).
\item Run Algorithm~\ref{alg: QSQN} for evaluating $(P,q(\ovl{x}))$ on $I$ within the time limit.
\item While $\tuples(\ans{q})$ contains less than $k$ tuples and the time limit was not reached yet, do:
  \begin{enumerate}
  \item Clear (empty) all the sets of the form $\tuples(\inp{p})$ and $\subqueries(\filter_{i,j})$. 
  \item Increase term-depth bound $l$ by 1.
  \item Run Algorithm~\ref{alg: QSQN} without Steps~1 and~2.
  \end{enumerate}
\item Return $\tuples(\ans{q})$.
\end{enumerate}

%-------------------------------------------------------------------------------

\section{Soundness and Completeness}
\label{section: sound-comp}

The following lemma states a property of Algorithm~\ref{alg: QSQN}. Its proof is straightforward.

\begin{lemma} \label{lemma: sound1}
Consider a run of Algorithm~\ref{alg: QSQN} (using parameter $l$) on a query $(P,q(\ovl{x}))$ and an \edb\ instance $I$ and let $\tuple{V,E,T,C}$ be the resulting QSQ-net. Let $v = \filter_{i,j}$ for some $1 \leq i \leq m$ and $1 \leq j \leq n_i$. Let $w = succ(v)$ and let $u = \filter_{i,j-1}$ if $j > 1$, and $u = \preFilter_i$ otherwise. Suppose that a subquery $(\ovl{s}',\delta')$ was transferred through $(v,w)$ at some step $k$. Then a subquery $(\ovl{s},\delta)$ was transferred through $(u,v)$ at some earlier step $h < k$ with the property that:
\begin{itemize}
\item if $\kind(v) = \edb$ and $\pred(v) = p$ then there exists $\tpc' \in I(p)$ such that $\atom(v)\delta$ is unifiable with a fresh variant of $p(\tpc')$ by an mgu $\gamma$, $\ovl{s}'=\ovl{s}\gamma$ and $\delta' = (\delta\gamma)_{|\postVars(v)}$
\item if $\kind(v) = \idb$ and $\pred(v) = p$ then there was $\tpc' \in \tuples(\ans{p})$ at step $k$ such that $\atom(v)\delta$ is unifiable with a fresh variant of $p(\tpc')$ by an mgu $\gamma$, $\ovl{s}'=\ovl{s}\gamma$ and $\delta' = (\delta\gamma)_{|\postVars(v)}$.
\myend
\end{itemize}
\end{lemma}

\begin{theorem}[Soundness] \label{theorem: sound}
After a run of Algorithm~\ref{alg: QSQN} on a query $(P,q(\ovl{x}))$ and an \edb\ instance $I$, for all intensional predicates $p$ of $P$, every computed answer $\tpc \in \tuples(\ans{p})$ is a correct answer in the sense that $P \cup I \models \V(p(\tpc))$.
\end{theorem}

\begin{proof}
We prove $P \cup I \models \V(p(\tpc))$ by induction on the number of the step at which $\tpc$ was added to $\tuples(\ans{p})$. Suppose $\tpc$ was added to $\tuples(\ans{p})$ as the result of transferring $\tpc$ through the edge $(\postFilter_i,\ans{p})$, which was triggered by the transfer of $(\tpc,\varepsilon)$ through the edge $(\filter_{i,n_i},\postFilter_i)$. 
Let $\ovl{s}_{n_i} = \tpc$ and $\delta_{n_i} = \varepsilon$. 
Let $v_0 = \preFilter_i$ and $v_j = \filter_{i,j}$ for $1 \leq j \leq n_i$. 
By Lemma~\ref{lemma: sound1}, for each $j$ from $n_i$ to $1$, there exists a subquery $(\ovl{s}_{j-1},\delta_{j-1})$ transferred through $(v_{j-1},v_j)$ such that:  
\begin{equation}\label{eq: HGWEQ}
\parbox{13cm}{if $\kind(v_j) = \edb$ and $\pred(v_j) = p_j$ then there exists $\tpc'_j \in I(p_j)$ such that $\atom(v_j)\delta_{j-1}$ is unifiable with a fresh variant of $p_j(\tpc'_j)$ by an mgu $\gamma_j$, $\ovl{s}_j=\ovl{s}_{j-1}\gamma_j$ and $\delta_j = (\delta_{j-1}\gamma_j)_{|\postVars(v_j)}$}
\end{equation}
\begin{equation}\label{eq: HHWEQ}
\parbox{13cm}{if $\kind(v_j) = \idb$ and $\pred(v_j) = p_j$ then there exists $\tpc'_j \in \tuples(\ans{p_j})$ such that $\atom(v_j)\delta_{j-1}$ is unifiable with a fresh variant of $p_j(\tpc'_j)$ by an mgu $\gamma_j$, $\ovl{s}_j = \ovl{s}_{j-1}\gamma_j$ and $\delta_j = (\delta_{j-1}\gamma_j)_{|\postVars(v_j)}$.}
\end{equation}

We have that $A_i\delta_0 = p(\ovl{s}_0)$. 
We prove by an inner induction on $1 \leq j \leq n_i + 1$ that:
\begin{equation}\label{eq: GHREX}
\parbox{11cm}{for every substitution $\theta$, if $P \cup I \models \V((B_{i,j} \land\ldots\land B_{i,n_i})\delta_{j-1}\theta)$ then $P \cup I \models \V(p(\ovl{s}_{j-1})\theta)$.}
\end{equation}

Base case ($j=1$): 
Since $P \cup I \models \V(\varphi_i)$, we have $P \cup I \models \V((B_{i,1} \land\ldots\land B_{i,n_i} \to A_i)\delta_0\theta)$. Hence, if $P \cup I \models \V((B_{i,1} \land\ldots\land B_{i,n_i})\delta_0\theta)$ then $P \cup I \models \V(A_i\delta_0\theta)$, which means $P \cup I \models \V(p(\ovl{s}_0)\theta)$.

\smallskip
Induction step: 
Suppose the induction hypothesis holds for $j \leq n_i$, i.e., 
\begin{equation}\label{eq: HJBRE}
\parbox{12cm}{for every $\theta$, if $P \cup I \models \V((B_{i,j} \land\ldots\land B_{i,n_i})\delta_{j-1}\theta)$ then $P \cup I \models \V(p(\ovl{s}_{j-1})\theta)$.}
\end{equation}
We show that it also holds for $j+1$, i.e.,  
\begin{equation}\label{eq: HJBRF}
\parbox{12cm}{for every $\theta'$, if $P \cup I \models \V((B_{i,j+1} \land\ldots\land B_{i,n_i})\delta_j\theta')$ then $P \cup I \models \V(p(\ovl{s}_j)\theta')$.}
\end{equation}
Suppose 
\begin{equation}\label{eq: KJTOS}
P \cup I \models \V((B_{i,j+1} \land\ldots\land B_{i,n_i})\delta_j\theta').
\end{equation}

Take $\theta = \gamma_j\theta'$. 

\begin{itemize}
\item Consider the case $\kind(v_j) = \edb$ and let $p_j = \pred(v_j)$. 
By~\eqref{eq: HGWEQ}, there exist $\tpc'_j \in I(p_j)$ and a fresh variant $\tpc''_j$ of $\tpc'_j$ such that $\gamma_j = mgu(B_{i,j}\delta_{j-1},p_j(\tpc''_j))$, $\ovl{s}_j=\ovl{s}_{j-1}\gamma_j$ and $\delta_j = (\delta_{j-1}\gamma_j)_{|\postVars(v_j)}$. We have $P \cup I \models \V(p_j(\tpc'_j))$, hence $P \cup I \models \V(p_j(\tpc''_j)\gamma_j)$, which means $P \cup I \models \V(B_{i,j}\delta_{j-1}\gamma_j)$. Hence $P \cup I \models \V(B_{i,j}\delta_{j-1}\gamma_j\theta')$, which means 
\begin{equation}\label{eq: HGREN}
P \cup I \models \V(B_{i,j}\delta_{j-1}\theta).
\end{equation}
Since $\delta_j = (\delta_{j-1}\gamma_j)_{|\postVars(v_j)}$ and $\theta = \gamma_j\theta'$, we have that 
\[ (B_{i,j+1} \land\ldots\land B_{i,n_i})\delta_j\theta' =
   (B_{i,j+1} \land\ldots\land B_{i,n_i})\delta_{j-1}\theta.
\]
This together with \eqref{eq: KJTOS}, \eqref{eq: HGREN} and \eqref{eq: HJBRE} implies $P \cup I \models \V(p(\ovl{s}_{j-1})\theta)$. Since $\ovl{s}_{j-1}\theta = \ovl{s}_{j-1}\gamma_j\theta' = \ovl{s}_j\theta'$, it follows that $P \cup I \models \V(p(\ovl{s}_j)\theta')$, which completes the proof of \eqref{eq: HJBRF} for the case $\kind(v_j) = \edb$. 

\item Consider the case $\kind(v_j) = \idb$ and let $p_j = \pred(v_j)$. 
By~\eqref{eq: HHWEQ}, there exist $\tpc'_j \in \tuples(\ans{p_j})$ and a fresh variant $\tpc''_j$ of $\tpc'_j$ such that $\gamma_j = mgu(B_{i,j}\delta_{j-1},p_j(\tpc''_j))$, $\ovl{s}_j=\ovl{s}_{j-1}\gamma_j$ and $\delta_j = (\delta_{j-1}\gamma_j)_{|\postVars(v_j)}$. By the inductive assumption of the outer induction, we have $P \cup I \models \V(p_j(\tpc'_j))$, hence $P \cup I \models \V(p_j(\tpc''_j)\gamma_j)$, which means $P \cup I \models \V(B_{i,j}\delta_{j-1}\gamma_j)$. Analogously as for the above case, we can derive that $P \cup I \models \V(p(\ovl{s}_j)\theta')$, which completes the proof of \eqref{eq: HJBRF} and \eqref{eq: GHREX}.
\end{itemize}

By \eqref{eq: GHREX}, when $j = n_i+1$ and $\theta = \varepsilon$, we have that $P \cup I \models \V(p(\ovl{s}_{n_i}))$, which means $P \cup I \models \V(p(\tpc))$.
\myend
\end{proof}

%-------------------------------------------------------------------------------

We need the following lemma for the completeness theorem. We assume that the sets of fresh variables used for renaming variables of input program clauses in SLD-refutations and in Algorithm~\ref{alg: QSQN} are disjoint. 

\begin{lemma} \label{lemma: comp1}
After a run of Algorithm~\ref{alg: QSQN} (using parameter $l$) on a query $(P,q(\ovl{x}))$ and an \edb\ instance $I$, for every \idb\ predicate $r$ of $P$, for every $\ovl{s} \in \tuples(\inp{r})$ and for every SLD-refutation of $P \cup I \cup \{\gets r(\ovl{s})\}$ that uses the leftmost selection function, does not contain any goal with term-depth greater than~$l$ and has a computed answer $\theta$ with the term-depth of $\ovl{s}\theta$ not greater than $l$, there exists $\ovl{s}' \in \tuples(\ans{r})$ such that $\ovl{s}\theta$ is an instance of a variant of~$\ovl{s}'$.
\end{lemma}

\begin{proof}
We prove this lemma by induction on the length of the mentioned SLD-refutation. Let $\theta_1,\ldots,\theta_h$ be the sequence of mgu's used in the refutation. We have that $r(\ovl{s})\theta_1\ldots\theta_h = r(\ovl{s})\theta$.
Suppose that the first step of the refutation of $P \cup I \cup \{\gets r(\ovl{s})\}$ uses an input program clause $\varphi'_i = (A'_i \gets B'_{i,1},\ldots,B'_{i,n_i})$, which is a variant of a program clause $\varphi_i = (A_i \gets B_{i,1},\ldots,B_{i,n_i})$ of $P$, resulting in the resolvent $\gets (B'_{i,1},\ldots,B'_{i,n_i})\theta_1$. Let $k_1 = 2$, $k_{n_i+1} = h+1$ and suppose that, for $1 \leq j \leq n_i$, 
\begin{equation}\label{eq: ABFED}
\parbox{11cm}{the fragment for processing $\gets B'_{i,j}\theta_1\ldots\theta_{k_j-1}$ of the refutation of $P \cup I \cup \{\gets r(\ovl{s})\}$ uses mgu's $\theta_{k_j},\ldots,\theta_{k_{j+1}-1}$.}
\end{equation}
Thus, after processing the atom $B'_{i,j-1}$ for $2 \leq j \leq n_i+1$, the next goal of the refutation of $\gets r(\ovl{s})$ is $\gets (B'_{i,j},\ldots,B'_{i,n_i})\theta_1\ldots\theta_{k_j-1}$. (If $j = n_i+1$ then the goal is empty.) 

Let $\varrho$ be a renaming substitution such that $\varphi'_i = \varphi_i\varrho$. Thus, $B'_{i,j} = B_{i,j}\varrho$ for $1 \leq j \leq n_i$. 
We can assume that $\varrho$ does not use any variable occurring in $\ovl{s}$. 

%-------------------------------------------------------------------------------

We will refer to the data structures used by Algorithm~\ref{alg: QSQN}. 
%Let $X = \Var(\ovl{s}) \cup \Var(\varphi_i)$. 

We first prove the following remark:

\begin{remark}\label{remark: HGDSX}
Let $1 \leq j \leq n_i$, $v = \filter_{i,j}$, $u = \filter_{i,j-1}$ if $j > 1$, and $u = \preFilter_i$ otherwise. If $(\tpc_{j-1},\delta_{j-1})$ is a subquery transferred through $(u,v)$ at some step and there exists a substitution $\eta$ such that 
\begin{equation} \label{eq: HGQSL}
(A_i,(B_{i,j},\ldots,B_{i,n_i}))\varrho\theta_1\ldots\theta_{k_j-1} = (r(\tpc_{j-1}),(B_{i,j},\ldots,B_{i,n_i})\delta_{j-1})\eta
\end{equation}
then there exist a subquery $(\tpc_j,\delta_j)$ transferred through $(v,succ(v))$ at some step and a substitution $\eta'$ such that 
\begin{equation} \label{eq: UYFDM}
(A_i,(B_{i,j+1},\ldots,B_{i,n_i}))\varrho\theta_1\ldots\theta_{k_{j+1}-1} = (r(\tpc_j),(B_{i,j+1},\ldots,B_{i,n_i})\delta_j)\eta'
\end{equation}
\end{remark}

Suppose the premises of this remark hold. 
Without loss of generality we assume that:
\begin{equation}\label{eq: LQYFS}
\parbox{12cm}{if ($\kind(v) = \edb$ and $T(v) = true$) or $\kind(v) = \idb$ then the subquery $(\tpc_{j-1},\delta_{j-1})$ was added to $\subqueries(v)$.}
\end{equation}

Since $B'_{i,j} = B_{i,j}\varrho$ and \eqref{eq: HGQSL}, we have that:
\begin{equation} \label{eq: LSJKF}
(\gets B'_{i,j}\theta_1\ldots\theta_{k_j-1}) = (\gets B_{i,j}\varrho\theta_1\ldots\theta_{k_j-1}) = (\gets B_{i,j}\delta_{j-1}\eta).
\end{equation}
Since the term-depth of $B_{i,j}\delta_{j-1}\eta = B'_{i,j}\theta_1\ldots\theta_{k_j-1}$ is not greater than $l$, the term-depth of $B_{i,j}\delta_{j-1}$ is also not greater than $l$. By \eqref{eq: ABFED}, \eqref{eq: LSJKF} and Lifting Lemma~\ref{lifting lemma}, we have that
\begin{equation} \label{eq: HGERI}
\parbox{12cm}{
there exists a refutation of $P \cup I \cup \{\gets B_{i,j}\delta_{j-1}\}$ using the leftmost selection function and mgu's $\theta'_{k_j},\ldots,\theta'_{k_{j+1}-1}$ such that the term-depths of goals are not greater than $l$ and $\eta\theta_{k_j}\ldots\theta_{k_{j+1}-1} = \theta'_{k_j}\ldots\theta'_{k_{j+1}-1}\mu$ for some substitution~$\mu$.
}
\end{equation}

Consider the case when the predicate $p = \pred(v)$ of $B_{i,j}$ is an extensional predicate.

Thus, 
\begin{equation} \label{eq: MDFRQ}
k_{j+1} = k_j + 1 
\end{equation}
and 
\begin{equation} \label{eq: DFUSW}
B_{i,j}\delta_{j-1}\theta'_{k_j} = p(\tpc')\sigma\theta'_{k_j} 
\end{equation}
where $p(\tpc')\sigma$ is the input program clause used for resolving $\gets B_{i,j}\delta_{j-1}$, with $\tpc' \in I(p)$ and $\sigma$ being a renaming substitution. 
Regarding the transfer of the subquery $(\tpc_{j-1},\delta_{j-1})$ through $(u,v)$, under the assumption~\eqref{eq: LQYFS}, Algorithm~\ref{alg: QSQN} unifies $\atom(v)\delta_{j-1} = B_{i,j}\delta_{j-1}$ with a fresh variant $p(\tpc')\sigma'$ of $p(\tpc')$, where $\sigma'$ is a renaming substitution, resulting in an mgu $\gamma$ (by~\eqref{eq: DFUSW}, $B_{i,j}\delta_{j-1}$ and $p(\tpc')\sigma'$ are unifiable, which is also justified below) and then transfers the subquery $(\tpc_{j-1}\gamma,(\delta_{j-1}\gamma)_{|\postVars(v)})$ through $(v,succ(v))$. 
Let 
\begin{equation} \label{eq: HJRTS}
\tpc_j = \tpc_{j-1}\gamma\;\;\textrm{and}\;\;
\delta_j = (\delta_{j-1}\gamma)_{|\postVars(v)}. 
\end{equation}

%-------------------------------------------------------------------------------

We have that $\sigma = \sigma'\sigma''$ for some renaming substitution $\sigma''$ such that 
\begin{equation} \label{eq: HGERP}
\textrm{$\sigma''$ does not use variables of $\tpc_{j-1}$, $\delta_{j-1}$ and $\preVars(v)$}.
\end{equation}
Thus $B_{i,j}\delta_{j-1}\sigma''\theta'_{k_j} = B_{i,j}\delta_{j-1}\theta'_{k_j}$, and by \eqref{eq: DFUSW} and the fact $\sigma = \sigma'\sigma''$, we have that
\[ (B_{i,j}\delta_{j-1})\sigma''\theta'_{k_j} = B_{i,j}\delta_{j-1}\theta'_{k_j} = p(\tpc')\sigma\theta'_{k_j} = (p(\tpc')\sigma')\sigma''\theta'_{k_j}.\] 
Hence, $B_{i,j}\delta_{j-1}$ and $p(\tpc')\sigma'$ are unifiable using $\sigma''\theta'_{k_j}$, while $\gamma$ is an mgu for them. Hence 
\begin{equation} \label{eq: HJERW}
\sigma''\theta'_{k_j} = \gamma\mu'
\end{equation}
for some substitution $\mu'$. 
Let $\eta' = \mu'\mu$. 
We have that:
\[
\begin{array}{ll}
& (A_i,(B_{i,j+1},\ldots,B_{i,n_i}))\varrho\theta_1\ldots\theta_{k_{j+1}-1} \\
= & ((A_i,(B_{i,j+1},\ldots,B_{i,n_i}))\varrho\theta_1\ldots\theta_{k_j-1})\theta_{k_j}\ldots\theta_{k_{j+1}-1} \\
= & (r(\tpc_{j-1}),(B_{i,j+1},\ldots,B_{i,n_i})\delta_{j-1})\eta\theta_{k_j}\ldots\theta_{k_{j+1}-1} \;\;(\mbox{by the assumption \eqref{eq: HGQSL}}) \\
= & (r(\tpc_{j-1}),(B_{i,j+1},\ldots,B_{i,n_i})\delta_{j-1})\theta'_{k_j}\ldots\theta'_{k_{j+1}-1}\mu \;\;(\mbox{by } \eqref{eq: HGERI}) \\
= & (r(\tpc_{j-1}),(B_{i,j+1},\ldots,B_{i,n_i})\delta_{j-1})\sigma''\theta'_{k_j}\ldots\theta'_{k_{j+1}-1}\mu \;\;(\textrm{by \eqref{eq: HGERP}}) \\
= & (r(\tpc_{j-1}),(B_{i,j+1},\ldots,B_{i,n_i})\delta_{j-1})\gamma\mu'\mu \;\;(\mbox{by } \eqref{eq: MDFRQ} \mbox{ and } \eqref{eq: HJERW}) \\
= & (r(\tpc_j),(B_{i,j+1},\ldots,B_{i,n_i})\delta_j)\eta' \;\;(\textrm{by \eqref{eq: HJRTS} and the fact $\eta' = \mu'\mu$}).
\end{array}
\]
We have shown~\eqref{eq: UYFDM} and thus proved Remark~\ref{remark: HGDSX} for the case when the predicate of $B_{i,j}$ is extensional.

%-------------------------------------------------------------------------------

Now consider the case when the predicate $p$ of $B_{i,j}$ is an intensional predicate.

By the assumption~\eqref{eq: LQYFS}, the subquery $(\tpc_{j-1},\delta_{j-1})$ was also added to $\unprocessedSQ(v)$.
Let $B_{i,j}\delta_{j-1} = p(\tpc')$. There must exist some tuple $\tpc$ more general than $\tpc'$ that was added to $\tuples(\inp{p})$ at some step. Let $\tpc\alpha = \tpc'$ for some substitution $\alpha$ that uses only variables from $\tpc$ and $\tpc'$. Thus,  
\begin{equation} \label{eq: GFYQS}
B_{i,j}\delta_{j-1} = p(\tpc)\alpha
\end{equation}

By \eqref{eq: HGERI} and Lifting Lemma~\ref{lifting lemma}, it follows that there exists a refutation of $P \cup I \cup \{\gets p(\tpc)\}$ using the leftmost selection function and mgu's $\theta''_{k_j},\ldots,\theta''_{k_{j+1}-1}$ such that the term-depths of the goals are not greater than $l$ and 
\begin{equation} \label{eq: AKFTS}
\alpha\theta'_{k_j}\ldots\theta'_{k_{j+1}-1} = \theta''_{k_j}\ldots\theta''_{k_{j+1}-1}\beta
\end{equation}
for some substitution $\beta$. 
By the inductive assumption, $\tuples(\ans{p})$ contains a tuple $\tpc''$ such that $\tpc\theta''_{k_j}\ldots\theta''_{k_{j+1}-1}$ is an instance of a variant of $\tpc''$. Since 
\[
\begin{array}{rcll}
B_{i,j}\delta_{j-1}\theta'_{k_j}\ldots\theta'_{k_{j+1}-1} & = & p(\tpc)\alpha\theta'_{k_j}\ldots\theta'_{k_{j+1}-1}\;\;\; & \textrm{(by \eqref{eq: GFYQS})} \\
& = & p(\tpc)\theta''_{k_j}\ldots\theta''_{k_{j+1}-1}\beta & \textrm{(by \eqref{eq: AKFTS}),} 
\end{array}
\]
it follows that 
\begin{equation} \label{eq: RHAKR}
B_{i,j}\delta_{j-1}\theta'_{k_j}\ldots\theta'_{k_{j+1}-1} \textrm{ is an instance of a variant of $p(\tpc'')$}.
\end{equation}

From certain moment there were both $(\tpc_{j-1},\delta_{j-1}) \in \subqueries(v)$ and $\tpc'' \in \tuples(\ans{p})$. Hence, at some step Algorithm~\ref{alg: QSQN} unified $\atom(v)(\delta_{j-1}) = B_{i,j}\delta_{j-1}$ with a fresh variant $p(\tpc'')\sigma$ of $p(\tpc'')$, where $\sigma$ is a renaming substitution. The atom $p(\tpc'')\sigma$ does not contain variables of $\tpc_{j-1}$, $\delta_{j-1}$, $\preVars(v)$ and $\theta'_{k_j}\ldots\theta'_{k_{j+1}-1}$. By~\eqref{eq: RHAKR}, $B_{i,j}\delta_{j-1}$ and $p(\tpc'')\sigma$ are unifiable. Let the resulting mgu be $\gamma$ and let 
\begin{equation} \label{eq: GKWTX}
\tpc_j = \tpc_{j-1}\gamma\;\;\textrm{and}\;\;
\delta_j = (\delta_{j-1}\gamma)_{|\postVars(v)}.
\end{equation}
Algorithm~\ref{alg: QSQN} then transferred the subquery $(\tpc_j,\delta_j)$ through $(v,succ(v))$.  

By \eqref{eq: RHAKR}, $B_{i,j}\delta_{j-1}\theta'_{k_j}\ldots\theta'_{k_{j+1}-1}$ is an instance of $p(\tpc'')\sigma$.
Let $\rho$ be a substitution with domain contained in $\Var(p(\tpc'')\sigma)$ such that $B_{i,j}\delta_{j-1}\theta'_{k_j}\ldots\theta'_{k_{j+1}-1} = p(\tpc'')\sigma\rho$. 
We have that 
\begin{equation} \label{eq: HREPS}
\parbox{12cm}{the domain of $\rho$ does not contain variables of $\tpc_{j-1}$, $\delta_{j-1}$, $\preVars(v)$ and $\theta'_{k_j}\ldots\theta'_{k_{j+1}-1}$}
\end{equation}
and $\theta'_{k_j}\ldots\theta'_{k_{j+1}-1} \cup \rho$ is a unifier for $B_{i,j}\delta_{j-1}$ and $p(\tpc'')\sigma$.
As $\gamma$ is an mgu for $B_{i,j}\delta_{j-1}$ and $p(\tpc'')\sigma$, we have that 
\begin{equation} \label{eq: GSKQX}
\gamma\mu' = (\theta'_{k_j}\ldots\theta'_{k_{j+1}-1} \cup \rho)
\end{equation}
for some substitution $\mu'$. 
Let $\eta' = \mu'\mu$.
We have that:
\[
\begin{array}{ll}
& (A_i,(B_{i,j+1},\ldots,B_{i,n_i}))\varrho\theta_1\ldots\theta_{k_{j+1}-1} \\
= & (r(\tpc_{j-1}),(B_{i,j+1},\ldots,B_{i,n_i})\delta_{j-1})\theta'_{k_j}\ldots\theta'_{k_{j+1}-1}\mu \;\;\textrm{(as shown before)} \\
= & (r(\tpc_{j-1}),(B_{i,j+1},\ldots,B_{i,n_i})\delta_{j-1})(\theta'_{k_j}\ldots\theta'_{k_{j+1}-1} \cup \rho)\mu \;\;\textrm{(by \eqref{eq: HREPS})} \\
= & (r(\tpc_{j-1}),(B_{i,j+1},\ldots,B_{i,n_i})\delta_{j-1})\gamma\mu'\mu \;\;\textrm{(by \eqref{eq: GSKQX})} \\
= & (r(\tpc_j),(B_{i,j+1},\ldots,B_{i,n_i})\delta_j)\eta' \;\;\textrm{(by \eqref{eq: GKWTX} and the fact $\eta' = \mu'\mu$)}.
\end{array}
\]
We have shown~\eqref{eq: UYFDM} and thus proved Remark~\ref{remark: HGDSX} for the case when the predicate of $B_{i,j}$ is intensional.
This completes the proof of this remark.

%-------------------------------------------------------------------------------

Recall that $r(\ovl{s})\varrho = r(\ovl{s})$. Since $\theta_1 = mgu(r(\ovl{s}),A'_i)$ and $A'_i = A_i\varrho$, it follows that $r(\ovl{s})\varrho\theta_1 = r(\ovl{s})\theta_1 = A'_i\theta_1 = A_i\varrho\theta_1$ and hence $\varrho\theta_1$ is a unifier for $r(\ovl{s})$ and $A_i$. Let $\gamma_0$ be the mgu Algorithm~\ref{alg: QSQN} uses to unify $r(\ovl{s})$ with $A_i$. Thus, $\gamma_0\eta_0 = \varrho\theta_1$ for some substitution $\eta_0$. Moreover, $(\tpc_0,\delta_0) = (\ovl{s}\gamma_0,(\gamma_0)_{|\preVars(\filter_{i,1})})$ is a subquery Algorithm~\ref{alg: QSQN} transferred through $(\preFilter_i,\filter_{i,1})$. Recall that $k_1 = 2$ and observe that the premises of Remark~\ref{remark: HGDSX} hold for $j = 1$ and for the subquery $(\tpc_0,\delta_0)$ using $\eta = \eta_0$. Hence there exist a subquery $(\tpc_1,\delta_1)$ transferred through $(\filter_{i,1},\filter_{i,2})$ at some step and a substitution $\eta_1$ such that 
\[ (A_i,(B_{i,2},\ldots,B_{i,n_i}))\varrho\theta_1\ldots\theta_{k_2-1} = (r(\tpc_1),(B_{i,2},\ldots,B_{i,n_i})\delta_1)\eta_1. \]

For each $1 < j \leq n_i$, we can apply Remark~\ref{remark: HGDSX} to obtain a subquery $(\tpc_j,\delta_j)$ and $\eta_j$ (for $\eta'$). Since $\postVars(\filter_{i,n_i}) = \emptyset$, it follows that, for $j = n_i$, we have that $(\tpc_{n_i},\varepsilon)$ is a subquery transferred through $(\filter_{i,n_i},\postFilter_i)$ at some step and 
\[ A_i\varrho\theta_1\ldots\theta_{k_{n_i+1}-1} = r(\tpc_{n_i})\eta_{n_i}.\]

Since $k_{n_i+1} = h+1$ and $\theta = (\theta_1\ldots\theta_h)_{|\Var(\ovl{s})}$, it follows that 
\[ r(\ovl{s})\theta = r(\ovl{s})\theta_1\ldots\theta_h
	= A'_i\theta_1\ldots\theta_h
	= A_i\varrho\theta_1\ldots\theta_h
	= r(\tpc_{n_i})\eta_{n_i}.
\]
Thus, $\ovl{s}\theta$ is an instance of $\tpc_{n_i}$. Since $(\tpc_{n_i},\varepsilon)$ was transferred through $(\filter_{i,n_i},\postFilter_i)$, $\tuples(\ans{r})$ will contain $\ovl{s}'$ such that $\tpc_{n_i}$ is an instance of a variant of $\ovl{s}'$. Clearly, $\ovl{s}\theta$ is also an instance of that variant of $\ovl{s}'$. This completes the proof. 
\myend
\end{proof}

%-------------------------------------------------------------------------------

\begin{theorem}[Completeness] \label{theorem: comp}
After a run of Algorithm~\ref{alg: QSQN} (using parameter $l$) on a query $(P,q(\ovl{x}))$ and an \edb\ instance $I$, for every SLD-refutation of \mbox{$P \cup I \cup \{\gets q(\ovl{x})\}$} that uses the leftmost selection function, does not contain any goal with term-depth greater than~$l$ and has a computed answer $\theta$ with term-depth not greater than $l$, there exists \mbox{$\ovl{s} \in \tuples(\ans{q})$} such that $\ovl{x}\theta$ is an instance of a variant of~$\ovl{s}$.
\myend
\end{theorem}

This theorem immediately follows from Lemma~\ref{lemma: comp1}. Together with Theorem~\ref{theorem: SLD soundness and completeness} (on completeness of SLD-resolution) it makes a relationship between correct answers of $P \cup I \cup \{\gets q(\ovl{x})\}$ and the answers computed by Algorithm~\ref{alg: QSQN} for the query $(P,q(\ovl{x}))$ on the \edb\ instance $I$. 

For queries and \edb\ instances without function symbols, we take term-depth bound $l = 0$ and obtain the following completeness result, which immediately follows from the above theorem.

\begin{corollary}%[Strong Completeness for the Case without Function Symbols]
\label{cor: comp}
After a run of Algorithm~\ref{alg: QSQN} using $l = 0$ on a query $(P,q(\ovl{x}))$ and an \edb\ instance $I$ that do not contain function symbols, for every computed answer $\theta$ of an SLD-refutation of $P \cup I \cup \{\gets q(\ovl{x})\}$ that uses the leftmost selection function, there exists $\tpc \in \tuples(\ans{q})$ such that $\ovl{x}\theta$ is an instance of a variant of~$\tpc$.
\myend
\end{corollary}

%-------------------------------------------------------------------------------
\section{Data Complexity}
\label{section: complexity}

In this subsection we estimate the {\em data complexity} of Algorithm~\ref{alg: QSQN}, which is measured w.r.t.\ the size of the \edb\ instance $I$ when the query $(P,q(\ovl{x}))$ and the term-depth bound $l$ are fixed.

If terms are represented as sequences of symbols or as trees then there will be a problem with complexity. Namely, unifying the terms $f(x_1,\ldots,x_n)$ and $f(g(x_0,x_0),\ldots,g(x_{n-1},x_{n-1}))$, we get a term of exponential length.\footnote{Another example is the pair $f(x_1,\ldots,x_n,x_1,\ldots,x_n)$ and $f(y_1,\ldots,y_n,g(y_0,y_0),\ldots,g(y_{n-1},y_{n-1}))$.} If the term-depth bound $l$ is used in all steps, including the ones of unification, then the problem will not arise. But we do not want to be so restrictive. 

To represent a term we use instead a rooted acyclic directed graph which is permitted to have multiple ordered arcs and caches nodes representing the same subterm. Such a graph will simply be called a DAG. As an example, the DAG of $f(x,a,x)$ has the root $n_f$ labeled by $f$, a node $n_x$ labeled by $x$, a node $n_a$ labeled by $a$, and three ordered edges outgoing from $n_f\,$: the first one and the third one are connected to $n_x$, while the second one is connected to $n_a$. 

The {\em size of a term $t$}, denoted by $size(t)$, is defined to be the size of the DAG of $t$ (i.e.\ the number of nodes and edges of the DAG of $t$). The sizes of other term-based expressions or data structures are defined as usual. For example, we define:
\begin{itemize}
\item the {\em size of a tuple $(t_1,\ldots,t_k)$} to be $size(t_1) + \ldots + size(t_k)$
\item the {\em size of a set of tuples} to be the sum of the sizes of those tuples
\item the {\em size of a substitution $\{x_1/t_1,\ldots,x_k/t_k\}$} to be $k + size(t_1) + \ldots + size(t_k)$
\item the {\em size of a node $v$ of a QSQ-net $\tuple{V,E,T,C}$} to be the sum of the sizes of the components of $C(v)$.
\end{itemize}

Using DAGs to represent terms, unification of two atoms $A$ and $A'$ can be done in polynomial time in the sizes of $A$ and $A'$. In the case $A$ and $A'$ are unifiable, the resulting atom and the resulting mgu have sizes that are polynomial in the sizes of $A$ and $A'$. Similarly, checking whether $A$ is an instance of $A'$ can also be done in polynomial time in the sizes of $A$ and~$A'$. 

The following theorem estimates the data complexity of Algorithm~\ref{alg: QSQN}, under the assumption that terms are represented by DAGs and unification and checking instances of atoms are done in polynomial time. 

\begin{theorem} \label{theorem: complexity}
For a fixed query and a fixed bound $l$ on term-depth, Algorithm~\ref{alg: QSQN} runs in polynomial time in the size of the \edb\ instance.
\end{theorem}

\begin{proof}
Consider a run of Algorithm~\ref{alg: QSQN} using parameter $l$ on a query $(P,q(\ovl{x}))$ and on an \edb\ instance $I$ with size $n$. Here, $(P,q(\ovl{x}))$ and $l$ are fixed. Thus, for every $1 \leq i \leq m$, $n_i$ is bounded by a constant. Similarly, if $p$ is an \idb\ predicate from $P$ then the arity of $p$ is also bounded by a constant.

Observe that the number of tuples that are added to any set of the form $\tuples(\inp{p})$ or $\tuples(\ans{p})$ are bounded by a polynomial of $n$. The reasons are: 
\begin{itemize}
\item \idb\ predicates come from $P$
\item constant symbols and function symbols come from $P$ and $I$
\item $\tuples(\inp{p})$ and $\tuples(\ans{p})$ consist of tuples with term-depth bounded by~$l$ 
\item a tuple is added to a set of the form $\tuples(\inp{p})$ or $\tuples(\ans{p})$ only when it is not an instance of a fresh variant of any tuple from the set
\item a tuple is deleted from a set of the form $\tuples(\inp{p})$ or $\tuples(\ans{p})$ only when it is an instance of a new tuple added to the set.
\end{itemize}
For similar reasons, the number of subqueries that are added to any set of the form $\subqueries(v)$ are also bounded by a polynomial of $n$. 

Consequently, the sizes of sets of the form $\tuples(\inp{p})$, $\tuples(\ans{p})$, $\subqueries(v)$, $\unprocessed(v,w)$, $\unprocessedTuples(v)$, $\unprocessedSubqueries(v)$ or $\unprocessedSQ(v)$ are bounded by a polynomial of $n$. 
Therefore, the size of the constructed QSQ-net is bounded by a polynomial of $n$, and any execution of procedure $\Transfer$, procedure $\Fire$ or function $\ActiveEdge$ is done in polynomial time in~$n$. 

A transfer or a firing for an edge $(u,v)$ is done only when a new tuple was added to $\tuples(u)$ or a new subquery was added to $\subqueries(u)$. Thus, we can conclude that Algorithm~\ref{alg: QSQN} runs in polynomial time in $n$.
\myend
\end{proof}

\begin{corollary}\label{cor: HGDSJ}
Algorithm~\ref{alg: QSQN} with term-depth bound $l = 0$ is a complete evaluation algorithm with PTIME data complexity for the class of queries over a signature without function symbols.
\myend
\end{corollary}

This corollary follows from Theorem~\ref{theorem: sound} (on soundness), Corollary~\ref{cor: comp} (on completeness) and the above theorem (on data complexity).

%-------------------------------------------------------------------------------

\section{QSQ-Nets with Tail Recursion Elimination}
\label{section: TRE}

A {\em query-subquery net structure with tail recursion elimination} (in short, {\em \QSQTRE-net structure}) of $P$ is a tuple $\tuple{V,E,T}$ defined similarly to a QSQ-net structure of $P$, but with the following differences:
\begin{itemize}
\item $T$ is a function, called the {\em type} of the net structure, mapping 
  \begin{itemize}
  \item each $\filter_{i,j} \in V$ such that the predicate of $B_{i,j}$ is extensional to {\em true} or {\em false}
  \item each \idb predicate to {\em true} or {\em false}. 
  \end{itemize}

\item If $A_i$ and $B_{i,n_i}$ have the same \idb predicate $p$ with $T(p) = true$ then $V$ does not contain the node $\postFilter_i$ and $E$ does not contain the edges $(\filter_{i,n_i},\postFilter_i)$, $(\postFilter_i,\ans{p})$ and $(\ans{p},filter_{i,n_i})$. 
\end{itemize}
The function $T$ can thus be called a memorizing type for extensional nodes $\filter_{i,j}$ (as in QSQ-net structures), and a tail-recursion-elimination type for intensional predicates. 

\begin{example}\label{example: HGDSB} 

Reconsider the positive logic program given in Example~\ref{example: HGDSA}:
\[
\begin{array}{l}
p(x,y) \gets q(x,y) \\
p(x,y) \gets q(x,z), p(z,y).
\end{array}
\]
A \QSQTRE-net structure $\tuple{V,E,T}$ of this program with $T(p) = true$ has the topological structure illustrated in Figure~\ref{fig: QSQTRE-net example}, which is like a loop.
\myend
\end{example}

\begin{figure}
\begin{footnotesize}
\begin{center}
\begin{tabular}{c}
\xymatrix{
&
*+[F]{\preFilter_1}
\ar@{->}[r]
&
*+[F]{\filter_{1,1}}
\ar@{->}[r]
&
*+[F]{\postFilter_1}
\ar@{->}[rrd]
\\
*+[F]{\inp{p}}
\ar@{->}[ru]
\ar@{->}[rd]
& & & & & 
*+[F]{\ans{p}}
\\
&
*+[F]{\preFilter_2}
\ar@{->}[r]
&
*+[F]{\filter_{2,1}}
\ar@{->}[r]
&
*+[F]{\filter_{2,2}}
\ar@/_{1.3pc}/@{->}[lllu]
} % \xymatrix
\end{tabular}
\end{center}
\end{footnotesize}
\caption{An illustration for Example~\ref{example: HGDSB}.
\label{fig: QSQTRE-net example}
}
\end{figure}
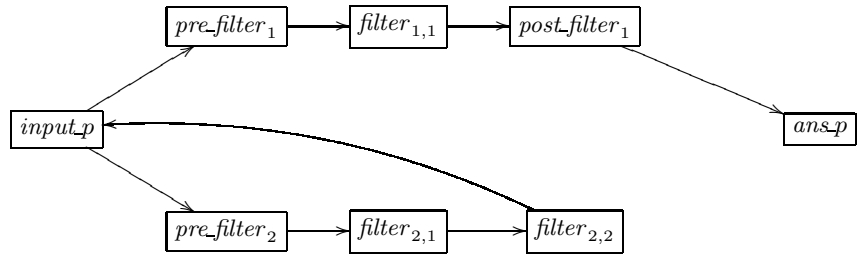

A {\em query-subquery net with tail recursion elimination} (in short, {\em \QSQTRE-net}) of $P$ is a tuple $N = \tuple{V,E,T,C}$ defined similarly to a QSQ-net of $P$, but with the following differences:
\begin{itemize}
\item $\tuple{V,E,T}$ is a \QSQTRE-net structure of $P$ 
\item if $v = \inp{p}$ and $T(p) = true$ then $C(v)$ consists of:
  \begin{itemize}
  \item $\tuplePairs(v)\,$: a set of pairs of generalized tuples of the same arity as $p$
  \item $\unprocessed(v,w)$ for $(v,w) \in E$: a subset of $\tuplePairs(v)$
  \end{itemize}
\item if $v = \filter_{i,n_i}$, $\kind(v) = \idb$, $\pred(v) = p$ and $T(p) = true$ then the structure $C(v)$ does not contain $\unprocessedSubqueries(v)$ and $\unprocessedTuples(v)$.
\end{itemize}

\begin{figure*}
\begin{procedure}[H]
\caption{transfer2($D,u,v$)\label{proc: TransferS}} 
\GlobalData{a Horn knowledge base $\tuple{P,I}$, a \QSQTRE-net $N = \tuple{V,E,T,C}$ of $P$, and a~term-depth bound $l$.}
\Input{data $D$ to transfer through the edge $(u,v) \in E$.}
%\DontPrintSemicolon

\BlankLine

\lIf{$D = \emptyset$}{\Return}\;

%\BlankLine

\uIf{$u$ is $\inp{p}$ and $T(p) = true$}{
$\Gamma := \emptyset$\;
\ForEach{$(\tpc,\tpc') \in D$}{
  \If{$p(\tpc)$ and $\atom(v)$ are unifiable by an mgu $\gamma$}{
  	$\AddSubquery(\tpc'\gamma, \gamma_{|\postVars(v)}, \Gamma, succ(v))$
  }
}
$\Transfer(\Gamma,v,succ(v))$
}
\uElseIf{$v$ is $\inp{p}$ and $T(p) = true$}{
  \ForEach{$(\tpc,\tpc') \in D$}{
     let $(\tpc_2,\tpc'_2)$ be a fresh variant of $(\tpc,\tpc')$\;
     \If{$(\tpc_2,\tpc'_2)$ is not an instance of any pair from $\tuplePairs(v)$}{
	\ForEach{$(\tpc_3,\tpc'_3) \in \tuplePairs(v)$}{
	   \If{$(\tpc_3,\tpc'_3)$ is an instance of $(\tpc_2,\tpc'_2)$}{
	      delete $(\tpc_3,\tpc'_3)$ from $\tuplePairs(v)$\;
	      \lForEach{$\tuple{v,w} \in E$}{delete $(\tpc_3,\tpc'_3)$ from $\unprocessed(v,w)$}
	   }
	}
	add $(\tpc_2,\tpc'_2)$ to $\tuplePairs(v)$\;
	\lForEach{$\tuple{v,w} \in E$}{add $(\tpc_2,\tpc'_2)$ to $\unprocessed(v,w)$}
     }
  }
}
\uElseIf{$v$ is $\filter_{i,n_i}$, $\kind(v) = \idb$, $\pred(v) = p$ and $T(p) = true$}{
  \ForEach{$(\tpc,\delta) \in D$}{
     \If{$\termDepth(\atom(v)\delta) \leq l$}{
        \If{no subquery in $\subqueries(v)$ is more general than $(\tpc,\delta)$}{
           delete from $\subqueries(v)$ all subqueries less general than $(\tpc,\delta)$\;
           delete from $\unprocessedSQ(v)$ all subqueries less general than $(\tpc,\delta)$\;
	   add $(\tpc,\delta)$ to both $\subqueries(v)$ and $\unprocessedSQ(v)$
        }
     }
  }
}
\Else{
Steps~\ref{transfer: step FDWES}-\ref{transfer: step FDWET} of procedure $\Transfer$ (given on page~\pageref{proc: Transfer}) with the recursive calls of $\Transfer$ replaced by calls of $\TransferS$
}
\end{procedure}

\smallskip

\begin{procedure}[H]
\caption{add-tuple-pair($\tpc,\tpc',\Gamma$)\label{proc: AddTuplePair}} 
\Purpose{add the pair of tuples $(\tpc,\tpc')$ to $\Gamma$, but keep in $\Gamma$ only the most general pairs.}
\BlankLine

let $(\tpc_2,\tpc'_2)$ be a fresh variant of $(\tpc,\tpc')$\;
\If{$(\tpc_2,\tpc'_2)$ is not an instance of any pair from $\Gamma$}{
   delete from $\Gamma$ all pairs that are instances of $(\tpc_2,\tpc'_2)$\;
   add $(\tpc_2,\tpc'_2)$ to $\Gamma$
}
\end{procedure}

\smallskip

\begin{procedure}[H]
\caption{compute-gamma()\label{proc: ComputeGamma}} 
\Purpose{a macro used in procedure $\FireS$ for replacing Step~\ref{fire: Step HGSAA} of procedure $\Fire$.}
\BlankLine

\uIf{$T(p) = false$}{
\lForEach{$(\tpc,\delta) \in \unprocessedSQ(u)$}{let $p(\tpc') = \atom(u)\delta$, $\AddTuple(\tpc',\Gamma)$}
}
\uElseIf{$j < n_i$}{
\lForEach{$(\tpc,\delta) \in \unprocessedSQ(u)$}{let $p(\tpc') = \atom(u)\delta$, $\AddTuplePair(\tpc',\tpc',\Gamma)$}
}
\Else{
\lForEach{$(\tpc,\delta) \in \unprocessedSQ(u)$}{let $p(\tpc') = \atom(u)\delta$, $\AddTuplePair(\tpc',\tpc,\Gamma)$}
}

\end{procedure}
\end{figure*}

The intuition behinds a pair $(\tpc,\tpc') \in \tuplePairs(\inp{p})$ is that: 
\begin{itemize}
\item $\tpc$ is a usual input tuple for $p$, but the intended goal at a higher level is $\gets p(\tpc')$
\item any correct answer for $P \cup I \cup \{\gets p(\tpc)\}$ is also a correct answer for $P \cup I \cup \{\gets p(\tpc')\}$
\item if a substitution $\theta$ is a computed answer of $P \cup I \cup \{\gets p(\tpc)\}$ then we will store in $\ans{p}$ the tuple $\tpc'\theta$ instead of $\tpc\theta$. 
\end{itemize}

Data transferred through an edge of the form $(\inp{p},v)$ or $(v,\inp{p})$ in a \QSQTRE-net $\tuple{V,E,T,C}$, where $p$ is an \idb predicate with $T(p) = true$, is redefined to be a finite set of pairs of generalized tuples of the same arity as $p$. 

We say that a tuple pair $(\tpc,\tpc')$ is {\em more general} than $(\tpc_2,\tpc'_2)$, and $(\tpc_2,\tpc'_2)$ is an {\em instance} of $(\tpc,\tpc')$, if there exists a substitution $\theta$ such that $(\tpc,\tpc')\theta = (\tpc_2,\tpc'_2)$.

Other notions for \QSQTRE-nets are defined similarly as for QSQ-nets.

\begin{figure*}[t]
\LinesNumberedHidden
\begin{algorithm}[H]
\caption{for evaluating a query $(P,q(\ovl{x}))$ on an \edb\ instance $I$.\label{alg: QSQTRE}}
%\DontPrintSemicolon

let $\tuple{V,E,T}$ be a \QSQTRE-net structure of $P$\tcp*{$T$ can be chosen arbitrarily}

set $C$ so that $N = \tuple{V,E,T,C}$ is an empty \QSQTRE-net of $P$\;

\BlankLine

let $\ovl{x}'$ be a fresh variant of $\ovl{x}$\;

\uIf{$T(q) = false$}{
  $\tuples(\inp{q}) := \{\ovl{x}'\}$\;
  \lForEach{$(\inp{q},v) \in E$}{$\unprocessed(\inp{q},v) := \{\ovl{x}'\}$}
}
\Else{
  $\tuplePairs(\inp{q}) := \{(\ovl{x}',\ovl{x}')\}$\;
  \lForEach{$(\inp{q},v) \in E$}{$\unprocessed(\inp{q},v) := \{(\ovl{x}',\ovl{x}')\}$}
}

\BlankLine

\While{there exists $(u,v) \in E$ such that $\ActiveEdge(u,v)$ holds}{
  select $(u,v) \in E$ such that $\ActiveEdge(u,v)$ holds\;
  \tcp{any strategy is acceptable for the above selection}
  $\FireS(u,v)$
}

\BlankLine
\Return $\tuples(\ans{q})$
\end{algorithm}
\end{figure*}

Procedure $\TransferS(D,u,v)$ (given on page~\pageref{proc: TransferS}) is a modified version of $\Transfer(D,u,v)$ for dealing with tail recursion elimination.

Let procedure $\FireS(u,v)$ be the modified version of $\Fire(u,v)$ obtained by:
\begin{itemize}
\item changing the calls of $\Transfer$ by calls of $\TransferS$ (with the same parameters)
\item replacing Step~\ref{fire: Step HGSAA} by macro $\ComputeGamma$ defined on page~\pageref{proc: TransferS}.
\end{itemize}

Algorithm~\ref{alg: QSQTRE} (given on page~\pageref{alg: QSQTRE}) is our reformulation of Algorithm~\ref{alg: QSQN} by using \QSQTRE-nets for evaluating queries. 

\begin{theorem}
Theorems~\ref{theorem: sound}, \ref{theorem: comp}, \ref{theorem: complexity} and Corollaries~\ref{cor: comp}, \ref{cor: HGDSJ} still hold when ``Algorithm~\ref{alg: QSQN}'' is replaced by ``Algorithm~\ref{alg: QSQTRE}''.
\myend
\end{theorem}

%-------------------------------------------------------------------------------

\section{Control Strategies}
\label{section: strategies}

Recall that in Algorithms~\ref{alg: QSQN} and \ref{alg: QSQTRE} we repeatedly select an active edge and fire the operation for it. Such selection is decided by the adopted control strategy, which can be arbitrary. 
In this section we describe two control strategies: the first one is to reduce the number of accesses to the secondary storage, while the second one is depth-first search, which gives priority to the order of clauses in the positive logic program defining intensional predicates and thus allows the user to control the evaluation to a certain extent. 

\subsection{Reducing the Number of Accesses to the Secondary Storage}

It is reasonable to assume that the computer memory is not large enough to load the whole extensional instance of the knowledge base into it and evaluation of queries cannot usually be done totally in the computer memory. Note that, not only extensional relations may be too large, but temporary relations used for computing intensional predicates like $\tuples(v)$, $\unprocessed(v,w)$, $\subqueries(v)$, \ldots\ may also be too large. Therefore, sometimes we have to load a relation into the computer memory, and sometimes we have to unload a relation to the secondary storage. As access to the secondary storage is time-consuming, it is desirable to reduce the total number of such accesses. Here is a strategy for this:
\begin{itemize}
\item If $(u,v)$ and $(u',v')$ are active edges of the considered QSQ-net/\QSQTRE-net and firing the edge $(u,v)$ can be done in the computer memory, while firing the edge $(u',v')$ requires loading some relations from the secondary storage then the edge $(u,v)$ has a higher priority than $(u',v')$ (for being selected).  
\item If firing any of edges $(u,v)$ and $(u',v')$ can be done in the computer memory then:
  \begin{itemize}
  \item the one that could enable a next operation be done in the computer memory (e.g.\ firing some edge $(v,w)$ or $(v',w')$) is considered to have a higher priority than the other
  \item if both of the edges are equal w.r.t.\ the above criterion then the one that could enable more next operations be done in the computer memory is considered to have a higher priority than the other
  \item if both of the edges are equal w.r.t.\ the above criteria then the one that processes more tuples/subqueries is considered to have a higher priority than the other.
  \end{itemize}

\item If no more operations can be done in the computer memory without loading relations from the secondary storage then select and load such a relation. The criteria for such selection are similar to the above mentioned ones. That is, we choose a relation to load into the computer memory that would enable more next operations be done in the computer memory and that would process more tuples/subqueries. 

\item If we want to load a relation into the computer memory but there is not enough available space in it then we have to select and unload an in-memory relation to the secondary storage. We can choose the in-memory relation that has not been used in the longest period to unload.
\end{itemize}

\subsection{Depth-First Evaluation}

The user may use Prolog programming style to specify the positive logic program defining intensional predicates. In such cases, e.g.\ as in Example~\ref{example1}, the order of the program clauses may be essential and depth-first search may increase efficiency of query evaluation. 

For each node of the considered QSQ-net/\QSQTRE-net we maintain and update its modification timestamp. For the depth-first evaluation approach, nodes are considered in the decreasing order of modification timestamps. When a node $v$ is considered, we choose an active edge $(v,w)$ to fire. If there is no such an edge, a next node in the mentioned order is chosen for consideration. If there are more than one successor $w$ of $v$ such that the edge $(v,w)$ is active, choose an edge $(v,w)$ according to the following strategy:
\begin{itemize}
\item If $v = \inp{p}$ then $w$ is the node $\preFilter_i$ with the smallest index $i$ such that $(v,w)$ is active (i.e.\ we consider the program clause $\varphi_i$ with the smallest index $i$ such that the edge $(v,\preFilter_i)$ is active). 

\item If $v = \filter_{i,j}$, $\kind(v) = \idb$, $v$ has two successors and both the edges $(v,succ(v))$ and $(v,succ_2(v))$ are active, then $w = succ(v)$.\footnote{If the net is with tail recursion elimination then $v$ may have only one successor.} 

\item If $v = \ans{p}$ then $w$ is the successor of $v$ with the biggest modification timestamp such that $(v,w)$ is active.
\end{itemize}

%-------------------------------------------------------------------------------

\section{Conclusions}
\label{section: conc}

We have provided the first framework for developing algorithms for evaluating queries to Horn knowledge bases with the properties that: 
the approach is goal-directed; each subquery is processed only once and each supplement tuple, if desired\footnote{when $T(v) = true$ for all nodes $v$ of the form $\filter_{i,j}$ with $\kind(v) = \edb$}, is transferred only once; operations are done set-at-a-time; and any control strategy can be used. 
The framework forms a generic evaluation method called QSQN. 
We have proved soundness and completeness of our generic evaluation method and showed that, when the term-depth bound is fixed, the method has PTIME data complexity.

This work is a continuation of~\cite{ToCL455}. It makes essential improvements: while the QSQR evaluation method of~\cite{ToCL455} uses iterative deepening search and does redundant recomputations, the QSQN evaluation method developed in this paper allows any control strategy and does not do redundant recomputations. The QSQN evaluation method is much more flexible, e.g., for reducing the number of accesses to the secondary storage.  

Our framework is an adaptation and a generalization of the QSQ approach of Datalog for Horn knowledge bases. One of the key differences is that we do not use adornments and annotations, but use substitutions instead. This is natural for the case with function symbols and without the range-restrictedness condition. When restricting to Datalog queries, it groups operations on the same relation together regardless of adornments and allows to reduce the number of accesses to the secondary storage although ``joins'' and ``projections'' would be more complicated. QSQ-nets are a more intuitive representation than the description of the QSQ approach of Datalog given in~\cite{AHV95}. Our notion of QSQ-net makes a connection to flow networks and is intuitive for developing efficient evaluation algorithms. For example, as shown in the paper, it is easy to incorporate tail recursion elimination into QSQ-nets, and as a result we have \QSQTRE-nets. 

In comparison with the most well-known evaluation methods, our QSQN evaluation method is more efficient than the QSQR evaluation method (as it does not do redundant recomputations) and is more flexible and thus has essential advantages over the bottom-up evaluation method based on magic-set transformation and improved seminaive evaluation (as shown in Example~\ref{example1}). 

%-------------------------------------------------------------------------------

%\bibliography{QSQN}
%\bibliographystyle{plain}

%-------------------------------------------------------------------------------

\end{document}